\documentclass[a4paper,onecolumn,11pt,unpublished]{quantumarticle}
\pdfoutput=1

\usepackage[utf8]{inputenc}
\usepackage[T1]{fontenc}
\usepackage[english]{babel}
\usepackage[numbers,sort&compress]{natbib}
\usepackage{hyperref}
\usepackage{qcircuit}
\usepackage{graphicx}
\usepackage{amsmath,amssymb,amsfonts}
\usepackage{amsthm}
\usepackage{algorithmicx}
\usepackage{algpseudocode}
\usepackage{tikz}
\usepackage{pgfplots}
\usepackage{booktabs}
\usepackage{tabularx}
\usepackage{float}
\pgfplotsset{compat=1.18}
\definecolor{methodclassic}{HTML}{6B6B6B}
\definecolor{methodknown}{HTML}{0072B2}
\definecolor{methodcount}{HTML}{009E73}
\definecolor{methodgrover}{HTML}{E69F00}
\definecolor{methodylc}{HTML}{CC79A7}
\usetikzlibrary{calc,positioning,decorations.pathreplacing,arrows.meta,shapes.geometric}
\AtBeginDocument{%
  \hypersetup{hypertexnames=false,urlcolor=quantumviolet}%
}%
\newcounter{algorithm}
\renewcommand{\thealgorithm}{\arabic{algorithm}}
\theoremstyle{plain}
\newtheorem{theorem}{Theorem}

\theoremstyle{definition}

\theoremstyle{remark}

\newcommand{\ket}[1]{\left|#1\right>}

\newcommand{\ketbra}[1]{\left|#1\right>\left<#1\right|}
\newcommand{\braket}[2]{\left<#1\mid#2\right>}

\newcommand{\norm}[1]{\lVert #1 \rVert_1}

\begin{document}

\title{Adaptive complete quantum search with unknown number of solutions}
\author{Raj Alexandru Guțoiu}
\author{Bogdan-Călin Ciobanu}
\author{Pantelimon George Popescu\texorpdfstring{\textsuperscript{*}}{}}
\email{pgpopescu@upb.ro}

\affiliation{Computer Science and Engineering Department,
  National University of Science and Technology POLITEHNICA Bucharest,
  Splaiul Independe\c{t}ei 313, (6), Bucharest, Romania }


\begin{abstract}
Unstructured database search is a fundamental problem that requires finding one or more elements satisfying a specific property within a set of candidates. While Grover's algorithm provides a quadratic speedup for finding a marked element when the number of solutions is known, finding all solutions when their number is unknown presents additional challenges. In this paper, we thoroughly analyze this problem by considering all the different regimes of combinations of classical, Grover and generalized fixed-point quantum search instances. We unify these into an adaptive Bayesian approach that maximizes the expected success probability per unit execution time at each searching instance. For a problem with a known number of solutions, much lower than the database size, we prove that our proposed approach is asymptotically optimal within the considered search model. Furthermore, we prove that for any probability distribution characterizing the likely number of solutions, in regimes where the classical verification time is at most $25\%$ of the quantum amplitude amplification time, hybrid classical and Grover approaches are preferred to generalized fixed-point. Outside of this regime, we derive an upper bound on the maximal improvement achievable by fixed-point search compared to the best Grover-based score. Additionally, we determine the optimal transition point between classical sampling and quantum searching using an integral approximation of the probability distribution of the number of solutions. The analytical results are also supported by numerical simulations of the proposed algorithm and by comparisons between the generalized fixed-point and Grover-based search policies.
\end{abstract}

\section{Introduction}

With the development of quantum computers, new algorithms obtaining improvements in time complexity can be implemented. In the case of search problems on unstructured databases, Grover's algorithm~\cite{grover1996fast} performs better than classical exhaustive search, reducing the complexity of finding one solution from the classical $\mathcal{O}(N)$ to the quantum $\mathcal{O}(\sqrt{N})$ when there is only one solution, and from $\mathcal{O}(N/M)$ to $\mathcal{O}(\sqrt{N/M})$ when there are $M$ solutions~\cite{boyer1998tight}. However, the original version requires knowledge of $M$ in order to select the appropriate number of Grover iterations and, because only an integer number of rotations can be performed, its success probability does not generally reach exactly one. Quantum amplitude amplification generalizes this procedure~\cite{brassard2002amplitude}, while suitable choices of arbitrary phases can produce an exact search when the initial marked fraction is known~\cite{hoyer2000arbitrary}. When $M$ is unknown, Boyer et al.~\cite{boyer1998tight} introduced a randomized trial-and-error procedure in which the number of Grover iterations is varied between successive searching attempts.

An important class of modifications of Grover's algorithm replaces the phase inversions of the oracle and the diffusion operator with more general phase rotations. The first phase-matching conditions showed that these rotations cannot generally be selected independently if amplitude amplification is to be preserved~\cite{long1999phase}. Arbitrary-phase amplitude amplification was subsequently formulated by H{\o}yer, who derived an exact matching condition and showed how a marked state can be obtained with certainty when the initial success probability is known~\cite{hoyer2000arbitrary}. Long proposed an analytical equal-phase construction with zero theoretical failure probability~\cite{long2001zero}, while more general matching conditions were later obtained for arbitrary initial states and transformations~\cite{long2002general}. Phase-matching search was further extended through different phase choices and sequences of multiple matched phases~\cite{li2007phase,toyama2008multiphase}.

These deterministic methods differ in the modifications they require from the original Grover circuit. Some of them assume that the oracle can apply an adjustable phase to the solutions, while restricted-oracle methods preserve the original phase-flip oracle and modify only the other reflections~\cite{roy2022deterministic,li2023adjustable}. 

When the number of solutions is unknown, another approach is to combine multiple quantum searching instances with classical processing between them. Boyer et al. introduced a trial-and-error strategy in which the number of Grover iterations is selected randomly from a progressively increasing interval~\cite{boyer1998tight}. Cade et al. later studied the practical, non-asymptotic performance of this strategy and introduced additional classical sampling and stopping mechanisms~\cite{cade2023quantifying}. Fixed-point search has also been combined with trial-and-error methods in order to obtain a high success probability without knowing the exact value of $M$~\cite{li2019hybrid}. More generally, partial Groverization can distribute a computation between classical and quantum procedures according to the available quantum resources, providing explicit trade-offs between runtime and the required quantum coherence time~\cite{eshaghian2025runtime}. Most of these methods terminate after finding one solution. The complete searching problem instead requires all the marked elements to be recovered and a criterion for deciding when no undiscovered solutions remain. Van Apeldoorn et al. presented an algorithm that finds all $M$ marked elements using the optimal $\mathcal{O}(\sqrt{NM})$ number of quantum queries, with a polylogarithmic gate overhead and a small quantum memory~\cite{vanapeldoorn2024multiple}. Lee and Nam considered the case in which $M$ must first be estimated and proposed a two-stage procedure that combines the estimation and discovery of the solutions~\cite{lee2024allsolutions}.

Bayesian inference provides a different mechanism for using the outcomes of previous searching attempts. One example, Particle-Guided Grover's Search~\cite{eshaghian2025hybrid}, represents prior information about the unknown number of solutions, updates it after unsuccessful Grover searches, and uses the resulting posterior to select the number of iterations in the next searching attempt. Particle filtering is used to approximate the distribution without storing all its possible values. The method is benchmarked against the modified trial-and-error search of Cade et al.~\cite{cade2023quantifying}, which is based on the unknown-$M$ search of Boyer et al.~\cite{boyer1998tight}. However, Particle-Guided Grover's Search is designed to find one marked element and updates the posterior using failed quantum searches. Bayesian methods have also been used to select circuit depths and process measurement outcomes in quantum amplitude estimation~\cite{koh2022foundations,ramoa2025bae,li2026biqae}, where the objective is to estimate an amplitude rather than to recover the identities of all marked elements. However, all these methods requiring both classical and quantum algorithms doesn't take into account the cost of classical measurements, nor the cost of testing the outcome of the quantum alogrithms.

In this paper, we consider the problem of finding all $M$ solutions of an unstructured search problem with $N$ possible candidates using a hybrid classical-quantum strategy when the value of $M$ is unknown. We explicitly account for the relative cost of candidate verification through the verification ratio $r$, defined as the execution time of one candidate verification divided by that of one quantum amplitude-amplification iteration. A quantum searching instance with $m$ amplification iterations therefore has scaled cost $m+r$, while a classical searching instance has cost $r$. We introduce an adaptive Bayesian complete-search policy that maintains a posterior distribution over the number of remaining solutions and updates it after every successful or unsuccessful searching instance. At each step, the algorithm selects between classical sampling, Grover search, and generalized fixed-point search by maximizing the posterior-predictive probability of finding a solution per unit cost. Tested candidates are removed from the active search space, and the search terminates when the posterior expected number of undiscovered solutions falls below a prescribed tolerance. We assume that the execution time of an amplitude-amplification iteration is dominated by the oracle evaluation, so the additional elementary-gate overhead required to prepare the evolving active-space superposition is neglected in the considered cost model.

For known $M=o(N)$, we prove that the greedy policy is asymptotically equivalent to the Bellman-optimal policy within this search model, with both selecting Grover searching instances whose success probability approaches $\sin^2\xi$, where $\xi\in(0,\pi/2)$ satisfies $\tan\xi=2\xi$. When $M\gg 1$ while $M/N\rightarrow0$, their expected scaled runtime approaches $1.38005\sqrt{NM}$. For unknown $M$, we prove that, for every posterior distribution and $0<r\leq1/4$, generalized fixed-point searching cannot improve the best greedy score obtainable from classical sampling and unmodified Grover search. For $r>1/4$, we derive a depth-dependent upper bound on the possible score improvement and a corresponding necessary upper bound on the additional implementation cost of a generalized amplitude-amplification iteration for such an improvement to remain possible. In the large-database approximation, we further characterize the transition from classical to quantum searching: for $1/4<r\leq1$, the first quantum score crossing is attained by an optimized depth-one searching instance, while for $r\rightarrow0^+$ and a fixed number $s$ of observed solutions the number of preceding classical queries satisfies $k_{\mathrm{tr}}(r,s)=C_s^2/r^2-2B_s/r+O_s(1)$. These analytical results are supported by decision-tree calculations and Monte Carlo simulations of the proposed complete-search policy.

\section{Background and related work}\label{sec:prelim}

Consider a search space of $N = 2^n$ elements and a Boolean function $f\colon \{0,\ldots,N-1\} \to \{0,1\}$. We assume that the function can be evaluated classically or implemented as a quantum oracle, allowing any element in the search space to be tested. The solution set $\mathcal{S} = \{x : f(x)=1\}$ has unknown cardinality $M = |\mathcal{S}|$, and we define the solution fraction as $\lambda = M/N$. The goal of complete search is to find every element of $\mathcal{S}$ while minimizing the expected total runtime. Since $M$ is unknown, the algorithm terminates according to a prescribed tolerance $\varepsilon$ on the expected number of solutions that remain undiscovered, as defined precisely in Sec.~\ref{sec:algorithm}.

Quantum algorithms for unstructured search rely on amplitude amplification to increase the probability of measuring a solution. In the oracle model considered here, such an algorithm returns a measured candidate $x$, but the measurement does not by itself reveal the value of $f(x)$, so an additional evaluation of $f(x)$, using either its classical implementation or the quantum oracle, is required. In the classical search considered here, a candidate $x$ is selected at random and $f(x)$ is evaluated; depending on whether the objective is to find one or multiple solutions, additional candidates may then be selected. Since an amplitude-amplification iteration requires an application of the oracle for $f$ and an additional operator, commonly called the diffuser, we expect one complete iteration to be more costly than a direct evaluation of $f(x)$. To quantify this, we define the testing ratio $0<r<1$ as the ratio between the time of one verification and the time of one amplitude-amplification iteration. The case $r=1$ is considered as the limiting case in which these costs become equal. When both classical and quantum implementations of $f$ are available for candidate verification, the faster implementation is used to determine the testing ratio.


\subsection{Grover's algorithm and generalizations}\label{sec:grover}

Grover's algorithm makes use of two operators, an oracle $\mathbf{O}$ that changes the phase of the marked states according to $f(x)$ and a diffuser $\mathbf{D}$. The combined action of these operators amplifies the amplitudes of the marked states and therefore increases the probability of measuring a solution~\cite{grover1996fast, boyer1998tight}. If any element of the database is equally likely to be a solution, the initial state is the uniform superposition $\ket{s} = \mathbf{H}^{\otimes n}\ket{0}^{\otimes n}$. The action of the two operators, $\mathbf{O}$ and $\mathbf{D}$, can be modeled in the space spanned by the states
$\ket{t} = \frac{1}{\sqrt{M}}\sum_{x\in\mathcal{S}}\ket{x}$ and $\ket{\bar t} = \frac{1}{\sqrt{N-M}}\sum_{x\notin\mathcal{S}}\ket{x}$, which are uniform superpositions of the marked and unmarked states, respectively. In this subspace, the effect of the two operators can be expressed as $\mathbf{O}=\mathbf{I}-2\ketbra{t}$ and $\mathbf{D}=2\ketbra{s}-\mathbf{I}$, respectively. One successive application of the two operators, $\mathbf{G}=\mathbf{D}\mathbf{O}$, is called a Grover iteration. By defining $\theta=\arcsin \sqrt{\lambda}$, the initial state can be expressed as $\ket{s} = \sin\theta\ket{t} + \cos\theta\ket{\bar t}$, and the state after $m$ Grover iterations is $\mathbf{G}^m\ket{s}=\sin((2m+1)\theta)\ket{t} + \cos((2m+1)\theta)\ket{\bar t}$. The probability of measuring a solution after $m$ iterations is therefore
\begin{equation}\label{eq:grover_prob}
  P_{\mathrm{Grover}}(\lambda,m) = \sin^2((2m+1)\theta).
\end{equation}

Due to this form of the success probability and its dependence on the solution fraction $\lambda$, there are two limitations to the original version of Grover's algorithm. First, knowledge of $\lambda$ is required to choose the number of iterations that maximizes the success probability. Second, even the optimal integer number of iterations does not generally guarantee a success probability of one. These limitations motivated strategies based on repeated searching instances, modifications of the initial state and its associated diffuser, modifications of the oracle and diffuser phases, or combinations of these approaches.

The generalisation of the initial state $\ket{s}$ is useful in the case of prior information about the distribution of the solutions~\cite{montanaro2011advice}. A custom state $\ket{s_k}$ can be prepared by a unitary $\mathbf{U}_k$ starting from the all-zero register state $\ket{0}\equiv\ket{0}^{\otimes n}$. We denote by $\mathbf{D}_0$ the diffuser associated with $\ket{0}$, $\mathbf{D}_0=2\ketbra{0}-\mathbf{I}$, and the generalized diffuser is then defined as $\mathbf{D}_k = \mathbf{U}_k\mathbf{D}_0\mathbf{U}_k^\dagger$. The original Grover algorithm is recovered when $\mathbf{U} = \mathbf{H}^{\otimes n}$. The change of the initial state and rotation basis for the diffuser changes the value of $\theta$ in eq.~\eqref{eq:grover_prob} to $\theta_k$, defined by $\sin^2\theta_k=\sum_{x\in\mathcal{S}}|\braket{x}{s_k}|^2$. When $\ket{s_k}$ is a uniform superposition over the active search space and $\ket{t}$ is defined using its marked elements, this reduces to $\sin\theta_k=\braket{t}{s_k}$. This generalization also allows previously tested elements to be removed from the active search space by assigning zero amplitude to their computational-basis states. Uniform superpositions over custom subsets can be prepared exactly using decision-diagram-based state-preparation methods, with the required resources depending on the structure of the subset~\cite{mozafari2021efficient,mozafari2022deterministic}.


Another generalization replaces the phase inversions implemented by the oracle and diffuser with arbitrary phase rotations. The phase-matching analysis of these rotations was introduced by Long et al.~\cite{long1999phase} and subsequently extended by Høyer~\cite{hoyer2000arbitrary}. In the notation used here, the generalized operators are:
\begin{equation}
  \mathbf{O}(\alpha)=\mathbf{I}-(1-e^{i\alpha})\ketbra{t},\qquad
  \mathbf{D}(\beta)=\mathbf{I}-(1-e^{i\beta})\ketbra{s}.
  \label{eq:generalized_operators}
\end{equation} 
For $\alpha=\beta=\pi$, $\mathbf{O}(\alpha)=\mathbf{O}$ and $\mathbf{D}(\beta)=-\mathbf{D}$, so the generalized iteration differs from the original Grover iteration only by a global phase. This generalizations allows for designing quantum algorithms that achieve a success probability of $100\%$ or manipulate the distribution of the success probability over a range of $\lambda$ in order to avoid overshooting the marked states. These algorithms are discussed in more detail in the following subsections.

\subsection{Exact quantum search}\label{sec:deterministic_search}

When $\lambda$ is known, exact quantum search can be achieved with at most one extra iteration compared to the one that maximizes the original Grover success probability.

Brassard et al.~\cite{brassard2002amplitude} generalized Grover search to amplitude amplification and showed that, when the initial success probability $a$ is known, a solution can be obtained with certainty while retaining a complexity of $\mathcal{O}(1/\sqrt{a})$. They presented two constructions: modifying the initial procedure to slightly reduce $a$, so that an integer number of standard iterations reaches the solution subspace exactly, or applying standard iterations followed by a final generalized iteration with appropriately chosen phases. For the uniform initial state considered here, $a=\lambda$.

By using the information related to a known element $x_0\notin \mathcal{S}$, the initial state can be modified by manipulating the amplitude of $\ket{x_0}$ such that the success probability of measuring a solution can achieve $100\%$~\cite{gutoiu2024simple}. The advantage of this method is not requiring to modify the oracle or diffuser phases, but it requires the information of a known non-solution element and the construction of an operator that prepares the modified initial state. 

By using the generalized oracle and diffuser described by eq.~\ref{eq:generalized_operators}, exact quantum search can be achieved without requiring a known nonsolution, provided that $\lambda$ is known. For a known $\lambda$, exact quantum search can be achieved using multiple strategies. If the same $\alpha$ and $\beta$ are used for every iteration, Høyer~\cite{hoyer2000arbitrary} derived the phase condition for exact quantum search, that is $\tan(\alpha/2)=(1-2\lambda)\tan(\beta/2)$, and showed how it can be combined with appropriate state and conditional-phase adjustments to obtain exact amplitude amplification. For the uniform initial state, a choice of phases for obtaining exact quantum search is given by:  
\begin{equation}
  \alpha=\beta=2\arcsin\left(\frac{\sin\left(\pi/(4m+2)\right)}{\sqrt{\lambda}}\right),
\end{equation}
which gives unit probability for any $m\geq m_{\mathrm{opt}}(\lambda)$~\cite{long2001zero}, where the optimal number of iterations is given by:
\begin{equation}\label{eq:m_opt}
  m_{\mathrm{opt}}(\lambda)=\left\lceil\frac{\pi}{4\arcsin\sqrt{\lambda}}-\frac{1}{2}\right\rceil.
\end{equation}
It was proven that this iteration count is the minimum number of iterations required to achieve a success probability of $100\%$~\cite{gutoiu2025optimal}, and, further more, if $m<m_{\mathrm{opt}}(\lambda)$, no generalization can obtain a success probability higher than the one provided by the original Grover's algorihm.

In the general case, the success probability as function of the number of iterations $m$ and $\lambda$ is given by~\cite{chowdhury2023phase}:
\begin{equation}\label{eq:fixed_phase_probability}
  \begin{aligned}
    P(\lambda,\alpha,\beta,m)
    ={}& \lambda
    +(1-2\lambda)(1-n_3^2)\sin^2(m\omega) \\
    &-2\sqrt{\lambda(1-\lambda)}
    \left[
      n_1n_3\sin^2(m\omega)
      +n_2\sin(m\omega)\cos(m\omega)
    \right].
  \end{aligned}
\end{equation}
where
\begin{equation*}
  \begin{aligned}
    \cos\omega
    &= \cos\left(\frac{\alpha-\beta}{2}\right)
       -2\lambda\sin\left(\frac{\alpha}{2}\right)
       \sin\left(\frac{\beta}{2}\right), \\[2pt]
    n_1
    &= \frac{2\sqrt{\lambda(1-\lambda)}
       \cos(\alpha/2)\sin(\beta/2)}{\sin\omega}, \\[2pt]
    n_2
    &= -\frac{2\sqrt{\lambda(1-\lambda)}
       \sin(\alpha/2)\sin(\beta/2)}{\sin\omega}, \\[2pt]
    n_3
    &= \frac{-\sin\left((\alpha-\beta)/2\right)
       -2\lambda\sin(\beta/2)\cos(\alpha/2)}{\sin\omega}.
  \end{aligned}
\end{equation*}

The optimal iteration count for deterministic search can also be attained when the oracle is restricted to the standard phase flip. For $\alpha=\pi$, the D2p construction alternates between two diffuser phases and, for $\lambda\leq1/4$, reaches unit success probability using $m_{\mathrm{opt}}(\lambda)$ oracle calls~\cite{roy2022deterministic,roy2023erratum}. This is at most one call more than the nearest-integer depth that maximizes the success probability of the original Grover algorithm. A more general framework allows either the oracle phase $\alpha$ or the diffuser phase $\beta$ to be specified arbitrarily, with the remaining parameters chosen to retain exact search~\cite{li2023adjustable}.

As a consequence, the maximum success probability of a searching algorithm, when the value of $\lambda$ is known, can be written as:
\begin{equation}
  P_{max}(\lambda, m)=
  \begin{cases}
    \sin^2((2m+1)\arcsin\sqrt{\lambda}),
      & m<m_{\mathrm{opt}}(\lambda),\\
    1, & m\geq m_{\mathrm{opt}}(\lambda).
  \end{cases}
  \label{eq:max_P}
\end{equation}

\subsection{Fixed-point search}\label{sec:fixed_point_search}

The arbitrary phase generalization of the oracle and diffuser operators, given by eq.~\eqref{eq:generalized_operators}, can also be used for searching algorithms that do not require the exact value of $\lambda$, being able to increase the success probability over a range of possible values of $\lambda$. However, such algorithms do not generally guarantee a success probability of $100\%$. 

For one iteration, the phase choice $\alpha=\beta=\pi/2$ gives $P(\lambda, \pi/2, \pi/2, 1)=\lambda(5-8\lambda+4\lambda^2)$, which is at least $25/27$ for any value of $\lambda\in[1/3, 1]$~\cite{li2007phase}. By using different phases at each iteration, a success probability of at least $0.988$ can be obtained for any value of $\lambda\in[0.1, 1]$ using six iterations with the phase relation $\alpha_j=\beta_{7-j}$~\cite{toyama2008multiphase}. 


Grover introduced the first fixed-point quantum search algorithm using the phase choice $\alpha=\beta=\pi/3$~\cite{grover2005fixed}. Instead of repeating the same generalized iteration, the transformation is applied recursively, with each level using the transformation obtained at the previous level and its inverse. Each recursion level cubes the failure probability, resulting after $k$ levels in $P_k(\lambda)=1-(1-\lambda)^{3^k}$ using $(3^k-1)/2$ oracle calls. Consequently, achieving a failure probability of at most $\delta$ requires $\mathcal{O}(\lambda^{-1}\log(1/\delta))$ oracle calls, and the quadratic speedup of Grover's algorithm is not retained.

A phase sequence that we will further refer to as the YLC phases was presented by Yoder et al.~\cite{yoder2014fixed}. For $L=2m+1$, let $w\in[0,1)$ be a parameter that determines the fixed-point interval. For the phase convention used in eq.~\eqref{eq:generalized_operators}, the phases are
\begin{equation}\label{eq:ylc_phases}
  \begin{aligned}
    \alpha_j
    &=2\operatorname{arccot}\left(
      w\tan\frac{(2j-1)\pi}{2m+1}
    \right),\\
    \beta_j
    &=-2\operatorname{arccot}\left(
      w\tan\frac{2j\pi}{2m+1}
    \right),
  \end{aligned}
\end{equation}
for $j=1,\ldots,m$. Here $\operatorname{arccot}$ takes values in $(0,\pi)$. The resulting success probability is
\begin{equation}\label{eq:ylc_chebyshev}
  P(\lambda,m,w)=1-\frac{T_L^2\left(\gamma\sqrt{1-\lambda}\right)}{T_L^2(\gamma)},
\end{equation}
where $\gamma=1/\sqrt{1-w^2}$ and $T_L$ is the Chebyshev polynomial of the first kind. At $w=0$, eq.~\eqref{eq:ylc_chebyshev} reduces to eq.~\eqref{eq:grover_prob}. For $\lambda\geq w^2$, the success probability is at least $1-1/T_L^2(\gamma)$.

For a target failure probability $\delta$ and a lower bound $\lambda_0\leq\lambda$, choose $w=\sqrt{\lambda_0}$ and let $L_\star$ be the smallest odd integer satisfying
\[
  L_\star\geq\frac{\operatorname{arcosh}(\delta^{-1/2})}{\operatorname{artanh}(\sqrt{\lambda_0})}.
\]
The corresponding query complexity is $\mathcal{O}(\log(1/\delta)/\sqrt{\lambda_0})$. If only the existence of at least one solution is assumed, $\lambda_0=1/N$ can be used.


\subsection{Quantum counting and amplitude estimation}\label{sec:counting}

Quantum counting was introduced by Brassard et al.~\cite{brassard1998counting} by combining Grover iterations with quantum phase estimation~\cite{kitaev1995quantum}. In the search subspace, the Grover operator $\mathbf{G}$ has eigenvalues $e^{\pm2i\theta}$, where $\sin^2\theta=\lambda=M/N$, and an inverse quantum Fourier transform~\cite{shor1997polynomial} provides an estimate of the corresponding eigenphase. One execution requires one final measurement, while the size of the control register determines its coherent oracle depth. For $M>0$ and a relative error $\eta$, Brassard et al.~\cite{brassard2002amplitude} showed that an estimate satisfying $|\widetilde{M}-M|\leq\eta M$ can be obtained with probability at least $2/3$ using an expected number of $\mathcal{O}(\eta^{-1}\sqrt{N/M})$ oracle evaluations. They also presented an exact-counting construction that determines the integer value of $M$ with probability at least $2/3$ using an expected number of $\mathcal{O}(\sqrt{(M+1)(N-M+1)})$ evaluations. This construction invokes the counting circuit three times and therefore requires three final measurements. Amplitude estimation extends quantum counting to an arbitrary state-preparation procedure, with the success probability given by $\lambda$ in the uniform search considered here.

Approximate counting and amplitude estimation can also be performed without phase estimation or controlled powers of $\mathbf{G}$. Aaronson and Rall~\cite{aaronson2020counting} adaptively select Grover depths and obtain a relative-error estimate of $M$ with probability at least $1-\delta$ using $\mathcal{O}(\eta^{-1}\sqrt{N/M}\log(1/\delta))$ oracle queries. Suzuki et al.~\cite{suzuki2020amplitude} introduced maximum-likelihood amplitude estimation, which combines repeated measurements performed at multiple amplification depths. Grinko et al.~\cite{grinko2021iterative} instead maintain a confidence interval for the amplitude and adaptively increase the circuit depth while avoiding phase ambiguity. These approaches use multiple measurements, and their query costs weight every circuit execution by its amplification depth rather than counting all measurements equally.

Counting and solution discovery do not always have to be performed separately. Lee and Nam~\cite{lee2024allsolutions} estimate $M$ using repeated Grover measurements and retain the distinct marked indices measured during this estimation stage. Their procedure therefore returns solutions while estimating their number, whereas phase-estimation-based quantum counting obtains its estimate from the measurement of the phase register.

\subsection{Search with an unknown number of solutions}\label{sec:unknown_search}

Boyer et al.~\cite{boyer1998tight} showed that, if a success probability of $100\%$ is not required in each attempt, the same number of Grover iterations can be repeated until the measured candidate is verified as a solution. Each attempt ends with one measurement of the search register and one verification of the measured candidate. For a fixed depth $m\geq1$, the expected number of Grover iterations is $m/P_{\mathrm{Grover}}(\lambda,m)$. In the sparse limit, the depth minimizing this quantity is determined by $\tan\xi=2\xi$,
with the numerical solution $\xi\approx1.16556$. 
The depth and success probability of one attempt then approach $m_\xi\sim0.58278/\sqrt{\lambda}$ and $P_\xi=\sin^2\xi\approx0.84458$, respectively, and the expected number of Grover iterations required to find a solution approaches $0.69003/\sqrt{\lambda}$.

When $\lambda$ is unknown, Boyer et al.~\cite{boyer1998tight} proposed a randomized trial-and-error procedure that selects $m$ uniformly from a progressively increasing range. Every unsuccessful searching instance increases the range and is followed by a new measurement and verification. The expected number of Grover iterations is $\mathcal{O}(\sqrt{N/M})=\mathcal{O}(1/\sqrt{\lambda})$ when $M>0$. For $M=0$, a timeout can be added to report the absence of solutions with one-sided error.

Brassard et al.~\cite{brassard2002amplitude} extended this strategy to general amplitude amplification through the QSearch algorithm. In the uniform search considered here, each trial first performs a direct attempt with success probability $\lambda$ and, if it fails, applies a randomly selected number of amplitude-amplification iterations from an increasing range. An unsuccessful trial therefore requires two measurements, while the successful trial requires one or two. For $\lambda>0$, QSearch has an expected complexity of $\mathcal{O}(1/\sqrt{\lambda})$; without an additional stopping rule, it does not terminate when $\lambda=0$.

Zalka~\cite{zalka1999unknown} presented a bounded procedure with one-sided error that tests successive hypotheses for $M$ and has query complexity $\mathcal{O}(\sqrt{N\log(1/\delta)})$. Cade et al.~\cite{cade2023quantifying} subsequently gave finite implementations of the randomized and deterministic procedures, including direct random tests and truncated quantum runs. Each Grover cycle ends with one measurement, while a truncated run can contain multiple cycles.

Okamoto and Watanabe~\cite{okamoto2001deterministic} presented a deterministic trial-and-error procedure in which the number of Grover iterations increases by a constant factor between attempts instead of being selected randomly. Li et al.~\cite{li2019hybrid} subsequently combined a deterministic growing-depth procedure with the YLC fixed-point sequence. On trial $k$, their algorithm first measures and tests the uniform superposition and, if this direct test fails, applies a YLC sequence with $m_k=\lceil c^{k-1}\rceil$ iterations followed by measurement and verification. Every unsuccessful trial therefore requires two measurements, while the successful trial requires one or two. The error-amplitude choice $0.5659$ and $c=1.523$ minimize the upper bound derived in their analysis. The expected query complexity is $\mathcal{O}(1/\sqrt{\lambda})$, with the bound reported as $5.643/\sqrt{\lambda}$ under their convention in which one arbitrary-phase Grover iteration uses two queries to the standard phase-$\pi$ oracle.

\subsection{Prior-guided and adaptive quantum search}\label{sec:bayes}

Prior information about the location of a marked element can be incorporated into the initial state or into the order in which subsets are searched. Montanaro~\cite{montanaro2011advice} searches geometrically increasing blocks ordered by their prior probabilities, while Huang and Pang~\cite{huang2023prior} and He et al.~\cite{he2024prior} optimize the initial amplitudes and the associated search transformation for a given prior distribution. Each search attempt ends with one computational-basis measurement, although an unsuccessful attempt may require the procedure to be repeated.

Particle-Guided Grover Search uses a prior distribution $p(\lambda)$ over the unknown solution fraction and updates it after every failed Grover searching instance~\cite{eshaghian2025hybrid}. The distribution is represented by $N_{\mathrm p}$ particles $\lambda_i$ with weights $p_i^{(k)}$ and updated using the failure likelihood $1-P_{\mathrm{Grover}}(\lambda_i,m)$ and Sequential Importance Resampling. The depth of the next searching instance is selected as
\[
  m^*
  =
  \operatorname*{arg\,max}_{1\leq m\leq\lceil\sqrt{N}\rceil}
  \frac{1}{m}
  \sum_{i=1}^{N_{\mathrm p}}
  p_i^{(k)}
  P_{\mathrm{Grover}}(\lambda_i,m).
\]
The selected depth maximizes the particle approximation of the expected success probability per Grover iteration. Each searching instance produces one measured candidate. After a failure, the particle weights are updated and a new depth is selected. For matching priors, the reported simulations obtain an average-query reduction of approximately $50\%$ compared with QSearch, while a mismatched prior can reduce this advantage.

Bayesian methods have also been used to select the parameters of amplitude-estimation circuits. Wang et al.~\cite{wang2021minimizing} select tunable phases according to the expected posterior variance, while Ramoa et al.~\cite{ramoa2025bae} represent the posterior using particles and adaptively select the amplification depth. These methods collect binary outcomes at the selected depths and determine the total number of measurements through their adaptive schedules, but their output is an amplitude estimate rather than the indices of marked states.

\subsection{Finding multiple or all marked elements}\label{sec:multiple_search}

If a one-solution search is repeated without modifying the oracle, solutions that have already been discovered may be measured again. Kessler et al.~\cite{kessler2023shots} analysed the resulting coupon-collector process for a fixed Grover depth. Every searching instance produces one measured database value, and the expected number of searching instances required to collect all $M$ solutions is
\begin{equation}\label{eq:coupon_expectation}
  \mathbb{E}[X_M]
  =\frac{M H_M}{P_{\mathrm{Grover}}(\lambda,m)},
\end{equation}
where $H_M$ is the $M$th harmonic number. They also derived the distribution and probability quantiles of the required number of searching instances.

Repeated measurements of previously discovered solutions can be avoided by modifying the oracle after every success. De Graaf and de Wolf~\cite{degraaf2002quantum} presented a certainty procedure given a known upper bound $M_{\max}$ on the number of marked elements. It performs $M_{\max}$ exact-search stages, each ending with one final measurement, and removes every returned solution from the active marked set. Its query complexity is $\mathcal{O}(\sqrt{NM_{\max}})$, which becomes the query-optimal $\mathcal{O}(\sqrt{NM})$ complexity when $M_{\max}=M$~\cite{vanapeldoorn2024multiple}.

When the indices already found are stored in a classical list $J$, a direct implementation of the modified oracle requires $\mathcal{O}(|J|\log N)$ additional gates for every membership test, giving a total gate complexity of $\widetilde{\mathcal{O}}(\sqrt{NM^3})$ despite retaining the optimal query complexity. Van Apeldoorn et al.~\cite{vanapeldoorn2024multiple} avoid this overhead by estimating $M$, sampling marked indices, partitioning the search space, and applying the certainty procedure within the resulting intervals. Their algorithm finds all marked indices with probability at least $1-\delta$ using $\mathcal{O}(\sqrt{NM})$ oracle queries and a polylogarithmic gate overhead. Its measurements occur in both the counting and sampling subroutines, for which the paper states query and gate bounds rather than a separate closed measurement bound.

Lee and Nam~\cite{lee2024allsolutions} proposed a two-stage procedure that estimates $M$ while collecting marked indices and then continues the discovery using this estimate, without modifying the oracle. The method assumes $M\leq\sqrt{N}$. In the first stage, one Grover iteration is applied before every measurement, and their experiments use $10\sqrt{N}$ measurements. The distinct marked elements observed during these measurements are stored in $\mathcal{S}_{found}$ and used to reduce the work of the second stage. The reported experiments find more than $80\%$ of the solutions during estimation. In the second stage, a depth selected from $\widetilde{M}$ is repeatedly applied, with one measurement after every searching instance, until a threshold number of consecutive measurements returns no new solution. The number of measurements considered in their analysis scales as $\mathcal{O}(M\log M)$, and the corresponding Grover-iteration scaling is $\mathcal{O}(\sqrt{NM}\log M)$. Their simulations recover more than $99\%$ of the solutions on average.

\section{Adaptive Bayesian complete quantum search}\label{sec:algorithm}

In this section we introduce an adaptive Bayesian strategy for complete search using both classical and quantum searching instances. The strategy maintains a posterior distribution over the number of remaining solutions and uses it to select each searching instance until a prescribed stopping tolerance is reached.

\subsection{Search state and admissible searching instances}

Let $M_0=M=|\mathcal{S}|$ be the true number of solutions of the search problem defined by $f$. After $k$ searching instances, let $\mathcal{S}_{\mathrm{tested}}$ denote the set of tested elements, with $|\mathcal{S}_{\mathrm{tested}}|=k$. The remaining search space has size $N_k=N-k$ and contains $M_k=|\mathcal{S}\setminus\mathcal{S}_{\mathrm{tested}}|$ solutions, corresponding to the solution fraction $\lambda_k=M_k/N_k$.

Each searching instance returns one candidate, which is then evaluated using $f$. Once a candidate has been tested, testing it again provides no additional information. We therefore remove every tested candidate from subsequent searching instances. The next searching instance uses the uniform superposition $\ket{s_k}=N_k^{-1/2}\sum_{x\notin\mathcal{S}_{\mathrm{tested}}}\ket{x}$ over the remaining elements.

The state $\ket{s_k}$ can be prepared by a unitary $\mathbf{U}_k$ satisfying $\mathbf{U}_k\ket{0}=\ket{s_k}$, and the generalized diffuser in eq.~\eqref{eq:generalized_operators} is then defined with respect to $\ket{s_k}$. The elementary-gate cost of implementing $\mathbf{U}_k$ and $\mathbf{U}_k^\dagger$ is not included in the cost assigned to a searching instance.

An admissible searching instance is specified by $m\in\mathbb{Z}_{\geq0}$ and $w\in[0,1)$. For $m\geq1$, the algorithm applies $m$ generalized amplitude-amplification iterations using the YLC phases given by eq.~\eqref{eq:ylc_phases}, measures a candidate, and evaluates $f$ to verify the result. Its success probability $P(\lambda_k,m,w)$ is given by eq.~\eqref{eq:ylc_chebyshev}. The choice $w=0$ recovers a Grover searching instance. The case $m=0$, for which we set $w=0$ by convention, denotes a classical searching instance in which a candidate is sampled uniformly from the remaining elements and $f$ is evaluated. In the probability model, this is equivalent to measuring $\ket{s_k}$ without applying any amplitude-amplification iterations, and its success probability is $P(\lambda_k,0,0)=\lambda_k$.

The cost of each searching instance is measured in units of one amplitude-amplification iteration. A quantum searching instance with $m$ iterations followed by verification has cost $m+r$, while a classical searching instance has cost $r$.

\subsection{Posterior distribution and Bayesian update}

Since the true value of $M_k$ is unknown, we describe the available information using the posterior distribution $\mathbf{p}^{(k)}=(p_0^{(k)},p_1^{(k)},\ldots,p_{N_k}^{(k)})$, where $p_i^{(k)}$ is the posterior probability that $M_k=i$. The distribution is defined over the support $\{0,1,\ldots,N_k\}$ and, before the outcome of the next searching instance is observed, acts as the prior distribution for the following Bayesian update.

Given the hypothesis $M_k=i$, the probability that a searching instance with parameters $(m,w)$ returns a solution is $P(i/N_k,m,w)$. When a successful outcome is observed, this conditional probability acts as the likelihood of the hypothesis $M_k=i$. For a failure, the corresponding likelihood is $1-P(i/N_k,m,w)$. We therefore define the scoring vector $\mathbf{s}(m,w,\mathbf{p}^{(k)})$ with elements
\begin{equation}\label{eq:score_vector}
  s_i(m,w,\mathbf{p}^{(k)})
  =p_i^{(k)}P(i/N_k,m,w).
\end{equation}
The elements of the scoring vector are the unnormalized posterior weights associated with a successful outcome. Their sum, $\norm{\mathbf{s}}=\sum_{i=0}^{N_k}s_i$, is the posterior-predictive probability that the searching instance returns a solution.

Let $\mathbf{s}^{(k)}=\mathbf{s}(m_k,w_k,\mathbf{p}^{(k)})$ denote the scoring vector for the searching instance selected at step $k$. The posterior distribution after removing the tested candidate is
\begin{equation}\label{eq:bayes_update}
  p_j^{(k+1)}=
  \begin{cases}
    \dfrac{s_{j+1}^{(k)}}{\norm{\mathbf{s}^{(k)}}},
      & \text{if a solution is found},\\[6pt]
    \dfrac{p_j^{(k)}-s_j^{(k)}}{1-\norm{\mathbf{s}^{(k)}}},
      & \text{if no solution is found}.
  \end{cases}
\end{equation}
where $j=0,\ldots,N_{k+1}$ and $N_{k+1}=N_k-1$. If a solution is found, the index is shifted because $M_{k+1}=M_k-1$. If no solution is found, the removed candidate is a nonsolution and $M_{k+1}=M_k$.
The same update rules apply to classical and quantum searching instances.

The initial distribution $\mathbf{p}^{(0)}$ represents the available prior information about the number of solutions. If every value $M_0\in\{0,1,\ldots,N\}$ is considered equally likely, then $p_i^{(0)}=1/(N+1)$. A known upper bound or other information about the number of solutions can instead be included by restricting or reweighting the support of $\mathbf{p}^{(0)}$.

\subsection{Searching instance selection}

The algorithm must minimize the expected scaled runtime. We denote by $\mathbb{E}[X\mid\mathbf{p}]$ the expected value of $X$ given the probability distribution $\mathbf{p}$, and by $\mathcal{T}_k$ the remaining scaled runtime after $k$ searching instances. For a searching instance with parameters $(m,w)$, let $\mathbf{p}_{\mathrm{S}}^{(k+1)}$ and $\mathbf{p}_{\mathrm{F}}^{(k+1)}$ be the posterior distributions obtained from eq.~\eqref{eq:bayes_update} after a success and a failure, respectively. The minimum expected remaining runtime satisfies the Bellman equation
\begin{equation}\label{eq:bellman_posterior}
  \begin{aligned}
  \mathbb{E}[\mathcal{T}_k\mid\mathbf{p}^{(k)}]
  =\min_{m,w}\Bigl\{
    &m+r
    +\norm{\mathbf{s}(m,w,\mathbf{p}^{(k)})}
      \mathbb{E}[\mathcal{T}_{k+1}\mid\mathbf{p}_{\mathrm{S}}^{(k+1)}]\\
    &+\left(1-\norm{\mathbf{s}(m,w,\mathbf{p}^{(k)})}\right)
      \mathbb{E}[\mathcal{T}_{k+1}\mid\mathbf{p}_{\mathrm{F}}^{(k+1)}]
  \Bigr\}.
  \end{aligned}
\end{equation}

The minimization is performed over the admissible searching instances, and the expected remaining runtime is zero for every posterior distribution that satisfies the stopping condition. For a general posterior distribution, each searching instance produces two possible posterior states, and different states can satisfy the stopping condition at different depths. The number of states can therefore grow exponentially, making an exact evaluation of the Bellman equation impractical for most posterior distributions. For posterior distributions with a simpler structure, the equation can still be evaluated. In particular, when the number of solutions is known, the posterior is a point mass and the state reduces to the discrete pair $(N_k,M_k)$, for which the Bellman equation can be solved using a finite table.

Another strategy for reducing the scaled runtime is to use a greedy policy that maximizes the posterior-predictive probability of finding a solution per unit cost. However, maximizing the score of each searching instance independently does not necessarily minimize the expected total runtime. Given the current posterior distribution, we define the score of a searching instance as
\begin{equation}\label{eq:score}
  S(m,w,r\mid\mathbf{p}^{(k)})
  =\frac{\norm{\mathbf{s}(m,w,\mathbf{p}^{(k)})}}{m+r},
\end{equation}
where $\mathbf{s}(m,w,\mathbf{p}^{(k)})$ is defined by eq.~\eqref{eq:score_vector}. For a classical searching instance, the score is $S(0,0,r\mid\mathbf{p}^{(k)})=\mathbb{E}[\lambda_k\mid\mathbf{p}^{(k)}]/r$, namely the posterior mean of the solution fraction in the remaining search space divided by the classical evaluation cost. Since $\norm{\mathbf{s}}\leq1$, the score satisfies $S(m,w,r\mid\mathbf{p}^{(k)})\leq1/(m+r)$. Therefore, once this upper bound is no greater than the best score already found, no larger value of $m$ can improve the result.

At step $k$, the greedy policy selects the admissible parameters that maximize the immediate score, $(m_k^*,w_k^*)=\underset{m,w}{\arg\max}S(m,w,r\mid\mathbf{p}^{(k)}).$ The implementation supports two posterior stopping rules. For the expected-remaining rule used by the complete-search benchmark, the algorithm terminates when the posterior mean number of remaining solutions does not exceed a user-defined tolerance $\varepsilon$,
\begin{equation}\label{eq:stop}
  \mathbb{E}\!\left[M_k\mid\mathbf{p}^{(k)}\right]
  =\sum_{i=0}^{N_k}i\,p_i^{(k)}
  \leq\varepsilon.
\end{equation}
The tolerance $\varepsilon$ is measured in units of solutions and quantifies the expected number of solutions remaining in the database. The alternative tail-probability rule terminates when $\Pr(M_k>0\mid\mathbf{p}^{(k)})=1-p_0^{(k)}\leq\delta$, where $\delta$ is a dimensionless probability threshold. The GPU experiments use this tail-probability rule, whereas the complete-search benchmark uses eq.~\eqref{eq:stop}. In the implementation, exhaustion of all active candidates and a configured maximum number of steps are additional terminal conditions.

\begin{figure}[H]
  \centering
  \refstepcounter{algorithm}
  \begin{minipage}[t]{0.59\textwidth}
    \vspace{0pt}
    \centering
    \textbf{(a) Algorithm~\thealgorithm: Adaptive Bayesian Complete Quantum Search}
    \label{alg:bayessearch}
    \vspace{0.3em}

    \raggedright
    \scriptsize
    \begin{algorithmic}[1]
      \Require $N$, oracle $f$, tolerance $\varepsilon$, testing ratio $r$, and $\mathbf{p}^{(0)}$
      \Ensure Set of found solutions $\mathcal{S}_{\mathrm{found}}$
      \State $\mathcal{S}_{\mathrm{found}}\gets\emptyset$;\ \ $\mathcal{S}_{\mathrm{tested}}\gets\emptyset$;\ \ $k\gets0$
      \While{$\mathbb{E}[M_k\mid\mathbf{p}^{(k)}]>\varepsilon$}
        \State $N_k\gets N-k$
        \State $(m_k^*,w_k^*)\gets\textsc{GreedyInstance}(\mathbf{p}^{(k)},N_k,r)$
        \State $x_k\gets\textsc{Search}(f,m_k^*,w_k^*,\mathcal{S}_{\mathrm{tested}})$
        \State $y_k\gets f(x_k)$
        \If{$y_k=1$}
          \State $\mathcal{S}_{\mathrm{found}}\gets\mathcal{S}_{\mathrm{found}}\cup\{x_k\}$
        \EndIf
        \State $\mathbf{p}^{(k+1)}\gets\textsc{BayesUpdate}(\mathbf{p}^{(k)},m_k^*,w_k^*,y_k,N_k)$
        \State $\mathcal{S}_{\mathrm{tested}}\gets\mathcal{S}_{\mathrm{tested}}\cup\{x_k\}$
        \State $k\gets k+1$
      \EndWhile
      \State \Return $\mathcal{S}_{\mathrm{found}}$
    \end{algorithmic}
  \end{minipage}
  \hfill
  \begin{minipage}[t]{0.38\textwidth}
    \vspace{0pt}
    \centering
    \textbf{(b) Flowchart}
    \vspace{0.3em}

    \resizebox{\linewidth}{!}{%
      \begin{tikzpicture}[
        block/.style={rectangle, draw, rounded corners, minimum width=2.8cm, minimum height=0.68cm, align=center, font=\scriptsize},
        decision/.style={diamond, draw, aspect=2.1, minimum width=2.7cm, align=center, font=\scriptsize, inner sep=1pt},
        arrow/.style={-{Stealth[length=1.8mm]}, thick},
        node distance=0.5cm and 0.8cm
        ]

        \node[block, fill=blue!10] (init) {Initialize $\mathbf{p}^{(0)}$,\\$\mathcal{S}_{\mathrm{found}}=\mathcal{S}_{\mathrm{tested}}=\emptyset$};
        \node[decision, below=of init, fill=yellow!15] (stop) {$\mathbb{E}[M_k\mid\mathbf{p}^{(k)}]\leq\varepsilon$?};
        \node[block, right=of stop, fill=blue!10] (done) {Return\\$\mathcal{S}_{\mathrm{found}}$};
        \node[block, below=of stop, fill=green!10] (select) {Select $(m_k^*,w_k^*)$\\by maximizing the score};
        \node[block, below=of select, fill=orange!10] (execute) {Execute the instance,\\measure $x_k$, and evaluate $f(x_k)$};
        \node[decision, below=of execute, fill=yellow!15] (found) {$f(x_k)=1$?};
        \node[block, right=of found, fill=purple!10] (record) {Add $x_k$ to\\$\mathcal{S}_{\mathrm{found}}$};
        \node[block, below=of found, fill=purple!10] (update) {Update $\mathbf{p}^{(k+1)}$;\\remove $x_k$ and set $k\gets k+1$};

        \draw[arrow] (init) -- (stop);
        \draw[arrow] (stop) -- node[above, font=\scriptsize] {yes} (done);
        \draw[arrow] (stop) -- node[right, font=\scriptsize] {no} (select);
        \draw[arrow] (select) -- (execute);
        \draw[arrow] (execute) -- (found);
        \draw[arrow] (found) -- node[above, font=\scriptsize] {yes} (record);
        \draw[arrow] (found) -- node[right, font=\scriptsize] {no} (update);
        \draw[arrow] (record) |- (update);
        \draw[arrow] (update.west) -- ++(-0.8,0) |- (stop.west);
      \end{tikzpicture}%
    }
  \end{minipage}
  \caption{Two representations of the adaptive Bayesian complete quantum search: (a) pseudocode and (b) the corresponding flowchart. At every nonterminal state, the algorithm selects the searching instance that maximizes the greedy score, evaluates the measured candidate, and updates the posterior distribution.}
  \label{fig:flowchart}
\end{figure}
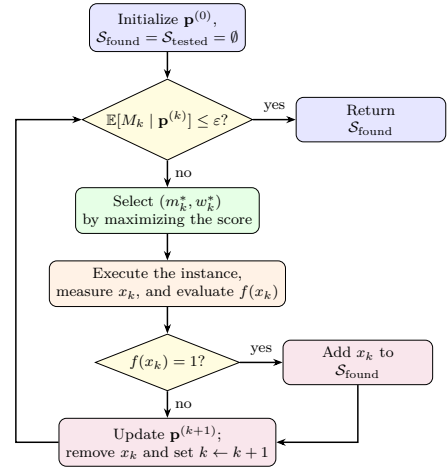

The complete procedure is presented as pseudocode and summarized by the accompanying flowchart in Fig.~\ref{fig:flowchart}. The complete algorithm starts from the initial distribution $\mathbf{p}^{(0)}$ and repeatedly selects the searching instance that maximizes the greedy score. The selected instance returns a candidate $x_k$, which is evaluated using $f$, added to $\mathcal{S}_{\mathrm{found}}$ if it is a solution, and removed from the remaining search space. The posterior distribution is then updated according to eq.~\eqref{eq:bayes_update}. This procedure is repeated until the stopping condition in eq.~\eqref{eq:stop} is satisfied. Since one candidate is removed after every searching instance, the algorithm terminates after at most $N$ instances. The returned set contains the solutions identified before termination, while $\varepsilon$ determines the tolerated posterior mean number of solutions that remain undiscovered.

\section{Theoretical characterization of searching-instance selection}\label{sec:theory}

This section analyzes how the posterior distribution and the testing ratio determine the selection of searching instances.

\subsection{Optimal hybrid search for known $M$}

If the number of solutions $M$ is known, the initial posterior distribution $\mathbf{p}^{(0)}$ is a point mass, with $p_M^{(0)}=1$ and zero everywhere else. After $k$ searching instances have found $s_k$ solutions, the number of remaining solutions is $M_k=M-s_k$, and the posterior remains a point mass, with $p_{M_k}^{(k)}=1$. For $m<m_{{opt}}(\lambda_k)$, the maximum success probability is attained by the original Grover algorithm with $w=0$. For $m\geq m_{{opt}}(\lambda_k)$, multiple implementations that achieve a success probability of $100\%$ exist, including a suitable choice of $w$ in the fixed-point algorithm. Since increasing $m$ beyond $m_{{opt}}(\lambda_k)$ increases the cost without increasing the success probability, the selection is restricted to a classical searching instance, Grover searching instances with $1\leq m<m_{{opt}}(\lambda_k)$, and an exact searching instance with $m=m_{{opt}}(\lambda_k)$. Their maximum success probability is given by eq.~\eqref{eq:max_P}.

Consequently, the maximum score attainable at a fixed depth $m$ is
\begin{equation}\label{eq:score_known_M}
  \max_{0\leq w<1}
  S(m,w,r\mid\mathbf{p}^{(k)})
  =
  \frac{P_{max}(\lambda_k,m)}{m+r},
\end{equation}

As the admissible iteration depths are therefore $\mathcal{A}_k=\{0,1,\ldots,m_{{opt}}(\lambda_k)\}$, the greedy policy selects the depth
\begin{equation}\label{eq:greedy_known_M}
  m_{{gr}}^*(N_k,M_k)
  =\underset{m\in\mathcal{A}_k}{\arg\max}\;
  \frac{P_{max}(\lambda_k,m)}{m+r}.
\end{equation}

Let $\mathbb{E}_{{B}}[\mathcal{T}\mid N_k,M_k]$ denote the minimum expected remaining scaled runtime from the state $(N_k,M_k)$ in the case of the greedy policy.

The Bellman equation in eq.~\eqref{eq:bellman_posterior} then reduces to
\begin{equation}\label{eq:bellman_known_M}
  \begin{aligned}
  \mathbb{E}_{{B}}[\mathcal{T}\mid N_k,M_k]
  =\min_{m\in\mathcal{A}_k}\Bigl\{&
    m+r
    +P_{max}(\lambda_k,m)
      \mathbb{E}_{{B}}[\mathcal{T}\mid N_k-1,M_k-1]\\
    &+\left[1-P_{max}(\lambda_k,m)\right]
      \mathbb{E}_{{B}}[\mathcal{T}\mid N_k-1,M_k]
  \Bigr\}.
  \end{aligned}
\end{equation}

The boundary conditions are $\mathbb{E}_{{B}}[\mathcal{T}\mid N_k,0]=0$ and $\mathbb{E}_{{B}}[\mathcal{T}_k\mid N_k,N_k]=rN_k$, where the second condition follows from the requirement that every returned element is individually verified. Let $m_{{B}}^*(N_k,M_k)$ be the number of iterations that minimizes $\mathbb{E}_{{B}}[\mathcal{T}\mid N_k,M_k]$. We denote by $\mathbb{E}_{{gr}}[\mathcal{T}\mid N,M]$ the expected scaled runtime of the greedy policy in eq.~\eqref{eq:greedy_known_M}. In Theorem~\ref{th:complexity_M_known} we prove that, if the number of solutions is low compared to the size of the database, then $\mathbb{E}_{{gr}}[\mathcal{T}\mid N,M]$ and $\mathbb{E}_{{B}}[\mathcal{T}\mid N,M]$ are asympthotically equal. 

\begin{theorem}\label{th:complexity_M_known}
Let $0<r\leq1$ be fixed and $M=o(N)$. The Bellman and greedy policies select Grover searching instances with $m_{{B}}^*\sim m_{{gr}}^*\sim\xi/(2\sqrt{\lambda_k})$ and success probability $P_\xi=\sin^2\xi$, where $\xi$ is the solution of $\tan \xi=2\xi$ in $(0,\pi/2)$. Their expected scaled runtimes satisfy
\begin{equation}\label{eq:known_M_asymptotic}
  \mathbb{E}_{{B}}[\mathcal{T}\mid N,M]
  \sim\mathbb{E}_{{gr}}[\mathcal{T}\mid N,M]
  \sim\frac{\xi\sqrt{N}}{2P_\xi}
  \sum_{j=1}^{M}\frac{1}{\sqrt{j}}.
\end{equation}
\end{theorem}

\begin{proof}
Define $\Delta_{{B}}(N_k,M_k)=\mathbb{E}_{{B}}[\mathcal{T}\mid N_k,M_k]-\mathbb{E}_{{B}}[\mathcal{T}\mid N_k,M_k-1]$.
Equation~\eqref{eq:bellman_known_M} is equivalent to
\[
  \mathbb{E}_{{B}}[\mathcal{T}\mid N_k,M_k]
  -\mathbb{E}_{{B}}[\mathcal{T}\mid N_k-1,M_k]\\
  =\min_{m\in\mathcal{A}_k}
  \left\{m+r-P_{max}(\lambda_k,m)
  \Delta_{{B}}(N_k-1,M_k)\right\}.
\]
For $M=o(N)$, the left-hand side is of lower order than $1/\sqrt{\lambda_k}$. Therefore,
\[
  \Delta_{{B}}(N_k,M_k)
  =\min_{m\in\mathcal{A}_k}
  \frac{m+r}{P_{max}(\lambda_k,m)}
  +o\!\left(\frac{1}{\sqrt{\lambda_k}}\right).
\]
This is the same minimization used in the greedy policy score described by eq.~\eqref{eq:greedy_known_M}. Let $\theta_k=\arcsin\sqrt{\lambda_k}$ and $u=(2m+1)\theta_k$. For $m<m_{{opt}}(\lambda_k)$,
\[
  \frac{m+r}{P_{max}(\lambda_k,m)}
  =\frac{u+(2r-1)\theta_k}{2\theta_k\sin^2u}.
\]
For fixed $0<r\leq1$ and $\theta_k\to0$, minimizing this expression is equivalent at leading order to minimizing $u/\sin^2u$ over $u\in(0,\pi/2]$, with the unique minimum on this interval at $u=\xi$, where $\tan\xi=2\xi$. Because $\xi<\pi/2$, the selected searching instance uses the original Grover phases and has success probability approaching $P_\xi$. For $j$ remaining solutions,
\[
  \frac{m_j^*+r}{P_{max}(j/N,m_j^*)}
  =\frac{\xi\sqrt{N}}{2P_\xi\sqrt{j}}+\mathcal{O}(1).
\]
Summing over $j=1,\ldots,M$ and using $\mathcal{O}(M)=o(\sqrt{NM})$ proves eq.~\eqref{eq:known_M_asymptotic}.
\end{proof}

Theorem~\ref{th:complexity_M_known} proves that, for a known number of solutions, the greedy policy provides asymthotical identical runtime as the optimization of the Bellman equation in the limit of low solution ratio. Further more, the constants $\xi$ and $P_\xi$ are the same results obtained by ~\cite{boyer1998tight} for finding only one solution of the searching problem with non-deterministic searching instances, proving that this strategy is asymthotically optimal also in the case of returning all solutions in the low solution ratio approximation. 

If $M\to\infty$ while $M/N\to0$, then
$\sum_{j=1}^{M}j^{-1/2}\sim2\sqrt{M}$, and
Eq.~\eqref{eq:known_M_asymptotic} becomes
\begin{equation}\label{eq:known_M_large_M}
  \mathbb{E}_{{B}}[\mathcal{T}\mid N,M]
  \sim
  \mathbb{E}_{{gr}}[\mathcal{T}\mid N,M]
  \sim
  \frac{\xi}{P_\xi}\sqrt{NM}
  \approx1.38005\sqrt{NM}.
\end{equation}

An exact searching instance at every stage has asymptotic scaled runtime $\pi\sqrt{NM}/2\approx1.57080\sqrt{NM}$. The Bellman and greedy expected costs therefore reduce the leading term by approximately $12.14\%$. The advantage of the Bellman policy is expected to become significant as the number of solutions increases and database shrinking has a stronger effect. In that regime, using shorter searching instances can reduce the remaining search space.

\subsection{Limits of the YLC fixed-point searching instances}

Because fixed-point generalizations introduce a continuous parameter $w$, which is optimized to maximize the score in eq.~\eqref{eq:score}, they add an extra layer of complexity to the complete-search problem. To simplify the analysis, we derive the transition conditions under which a choice $w\neq 0$ improves the score of a search instance and also quantify the improvement in the greedy score compared to the one of the unmodified Grover's algorithm. For this, we define the best score obtained by Grover's algorithm over a posterior distribution $\mathbf{p}$ as $S^*_{{Grover}}=\max_{m\in\mathbb{Z}_{\geq0}}S(m,0,r\mid\mathbf{p})$, including the score of the classical search in the case of $m=0$.

Given a searching algorithm described by the success probability distribution as function of the solution ratio, $\mathbf{P}(\theta)$, with $\theta=\arcsin(\sqrt{\lambda})$, we define its associated failure amplitude as $\mathbf{P}(\theta)=1-|A(\theta)|^2$, noting that multiple failure amplitudes can obtain the same success probability distribtuion. The main properties characterizing $A(\theta)$ will then be $A(0)=1$, $A(\pi/2)=0$ and $A(\theta)$ can be written as a polynomial of degree at most $2m+1$ in $e^{i\theta}$. Given the Laurent expansion of the failure amplitude:
\begin{equation}\label{eq:positive_fourier_amplitude}
  A(\theta)=\sum_{j=0}^{L}\pi_j e^{i(2j-L)\theta},
\end{equation}
the condition $A(0)=1$ will impose $|\sum_{j=0}^{L}\pi_j|=1$. Also, for any phase-modified search algorithm, there exists a failure amplitude $A(\theta)$ such that $A(\theta)/\cos(\theta)$ is a polynomial in $\sin^2\theta$ of degree at most m ~\cite{rall2023amplitude}. This implies that $A(-\theta)=A(\theta)$, so eq.\eqref{eq:positive_fourier_amplitude} can be restricted to $\pi_j=\pi_{L-j}$.

The YLC failure amplitude is $A_{{YLC}}(\theta)=T_L(\gamma\cos\theta)/T_L(\gamma)$, so that $P(\lambda,m,w)=1-|A_{{YLC}}(\theta)|^2$. The Chebyshev polynomial $T_L(\gamma\cos\theta)$ can be obtained from the trace of the matrix $\mathbf{B}_{\theta}^L(\gamma)$, where
\begin{equation}\label{eq:chebyshev_transfer_matrix}
    \mathbf{B}_{\theta}(\gamma)
   =\begin{pmatrix}
     \gamma e^{i\theta} & \sqrt{\gamma^2-1} e^{i\theta}\\
      \sqrt{\gamma^2-1} e^{-i\theta} & \gamma e^{-i\theta}
   \end{pmatrix},
\end{equation}
Taking the trace of the $L$th power in eq.~\eqref{eq:chebyshev_transfer_matrix} gives $\operatorname{tr}(\mathbf{B}_{\theta}^L(\gamma))=2T_L(\gamma\cos\theta)$, which can be expanded using coefficients $b_j(\gamma)\geq0$:
\begin{equation}\label{eq:chebyshev_expansion}
    2T_L(\gamma\cos\theta)=\sum_{j=0}^{L}b_j(\gamma) e^{i(2j-L)\theta}.
\end{equation}
Setting $\theta=0$ in eq.~\eqref{eq:chebyshev_expansion} gives $\sum_{j=0}^{L}b_j(\gamma)=2T_L(\gamma)$. We may therefore define $\pi_j(\gamma)=b_j(\gamma)/[2T_L(\gamma)]$, for which $\sum_{j=0}^{L}\pi_j=1$, and write the failure amplitude $A_{{YLC}}(\theta)$ according to eq.\eqref{eq:positive_fourier_amplitude}.

The result provided in Theorem~\ref{th:phase_score_bound} proves that, for any $0<r\leq1/4$, a searching instance that accepts a failure amplitude $A(\theta)$ whose Laurent expansion factors are positive, cannot improve the best score obtained by unmodified Grover's algorithm or classical search. 

\begin{theorem}\label{th:phase_score_bound}
For every posterior distribution $\mathbf{p}^{(k)}$ and $0<r\leq1/4$, any failure amplitude $A(\theta)=\sum_{j=0}^{L}\pi_j e^{i(2j-L)\theta}$, with $\pi_j\geq 0$ $\forall j\in\{0, 1,\ldots, L\}$ obeys
\begin{equation}\label{eq:phase_score_bound}
  \frac{\displaystyle
    \sum_{i=0}^{N_k}p_i^{(k)}
    \left[1-\left|A\!\left(\arcsin\sqrt{i/N_k}\right)\right|^2\right]}
  {m+r}
  \leq S^*_{{Grover}}.
\end{equation}
The constant $1/4$ is sharp.
\end{theorem}

\begin{proof}
For the fixed posterior $\mathbf{p}^{(k)}$, we define the auxiliary moments
\begin{equation}\label{eq:aux_moments}
  g_j
  =\sum_{i=0}^{N_k}p_i^{(k)}
  \sin^2\!\left(j\arcsin\sqrt{\frac{i}{N_k}}\right),
\end{equation}
and by using the inequality:
\begin{equation*}
  \sin^2(2\ell\theta)
  \leq
  \frac{\sin^2((2\ell-1)\theta)+\sin^2((2\ell+1)\theta)}{2}
  +\sin^2\theta,
\end{equation*}
we get that
 \begin{equation*}
 g_{2\ell}\leq\frac{g_{2\ell-1}+g_{2\ell+1}}{2}+g_1.
 \end{equation*}
Since $g_{2\ell+1}=\norm{\mathbf{s}(\ell,0,\mathbf{p}^{(k)})}$, eq.~\eqref{eq:score} gives $g_{2\ell+1}\leq(\ell+r)S^*_{{Grover}}$. Combined with the preceding inequality, this provides the additional inequality $g_{2\ell}\leq(\ell+2r-1/2)S^*_{{Grover}}$. Therefore,
\begin{equation}\label{eq:moment_coefficient_bound}
  g_j\leq c_r(j)S^*_{{Grover}},
\end{equation}
where
\begin{equation}\label{eq:cr_definition}
  c_r(j)=
  \begin{cases}
    0, & j=0,\\
    (j-1)/2+r, & j\geq1\text{ odd},\\
    j/2+2r-1/2, & j\geq2\text{ even}.
  \end{cases}
\end{equation}

Let $Z$ and $Z'$ be independent random variables with $\Pr(Z=j)=\Pr(Z'=j)=\pi_j$, where $\pi_j$ are the positive Fourier coefficients in the decomposition of $A(\theta)$ given by eq.~\eqref{eq:positive_fourier_amplitude}, and let $J=|Z-Z'|$. Then $|A(\theta)|^2=\sum_{j=0}^{L}\sum_{k=0}^{L}\pi_j\pi_k e^{2i(j-k)\theta}=\mathbb{E}[\cos(2(Z-Z')\theta)]$. Since $\cos(2x)=1-2\sin^2x$ and $\sin^2x$ is even, this gives $1-|A(\theta)|^2=2\mathbb{E}[\sin^2(J\theta)]$.

Consequently,
\begin{align*}
  \sum_{i=0}^{N_k}p_i^{(k)}
  \left[1-\left|A(\theta_i)\right|^2\right]
  &=2\sum_{i=0}^{N_k}p_i^{(k)}
  \mathbb{E}\!\left[\sin^2(J\theta_i)\right] \\
  &=2\mathbb{E}\!\left[
    \sum_{i=0}^{N_k}p_i^{(k)}\sin^2(J\theta_i)
  \right] \\
  &=2\mathbb{E}[g_J].
\end{align*}
Together with eq.~\eqref{eq:moment_coefficient_bound}, this gives
\begin{equation}\label{eq:phase_moment_reduction}
  \sum_{i=0}^{N_k}p_i^{(k)}
  \left[1-\left|A(\theta_i)\right|^2\right]
  =2\mathbb{E}[g_J]
  \leq2S^*_{{Grover}}\mathbb{E}[c_r(J)].
\end{equation}
For $r=0$, $2c_0(J)=\max\{J-1,0\}$. Since $0\leq J\leq L$ and two independent random variables supported on an interval of length $L$ satisfy $\mathbb{E}|Z-Z'|\leq L/2$,
\[
  2\mathbb{E}[c_0(J)]
  \leq\frac{L-1}{L}\mathbb{E}[J]
  \leq m.
\]
For $r=1/4$, we have
\[
  2c_{1/4}(J)
  \leq
  |2c_{1/4}(Z)-2c_{1/4}(Z')|.
\]
The random variables $2c_{1/4}(Z)$ and $2c_{1/4}(Z')$ are independent and identically distributed, with support contained in $[0,L-1/2]$. Therefore, the same mean-distance bound gives
\[
  2\mathbb{E}[c_{1/4}(J)]
  \leq\frac{L-1/2}{2}
  =m+\frac14.
\]
Both sides of $2\mathbb{E}[c_r(J)]\leq m+r$ are affine in $r$. Interpolation between $r=0$ and $r=1/4$ therefore proves eq.~\eqref{eq:phase_score_bound}.

To prove sharpness, fix an odd $m\geq3$ and define $U_\ell(\vartheta)=\sin^2((2\ell+1)\vartheta)-\sin^2\vartheta$. There exists $\vartheta\in(\pi/[4(2m-1)],\pi/(4m+2))$ such that $mU_{m-1}(\vartheta)=(m-1)U_m(\vartheta)$. Indeed, $H_m(\vartheta)=mU_{m-1}(\vartheta)-(m-1)U_m(\vartheta)$ changes sign at the endpoints of this interval. To compare this depth with every other depth, let $h=2\vartheta$ and $D_\ell=\ell[\cos h-\cos((2m+1)h)]-m[\cos h-\cos((2\ell+1)h)]$. The balancing condition gives $D_0=D_{m-1}=D_m=0$, while $D_{\ell+2}-2D_{\ell+1}+D_\ell=-4m\sin^2h\cos((2\ell+3)h)$.
The changes between discrete concavity and convexity determined by the sign of the cosine, together with the bound $\cos h-\cos((2m+1)h)>1$ for the remaining tail, give $D_\ell\geq0$ for every $\ell$. Since $D_\ell=2[\ell U_m(\vartheta)-mU_\ell(\vartheta)]$, it follows that
\begin{equation}\label{eq:sharpness_balanced_depth}
  mU_\ell(\vartheta)\leq\ell U_m(\vartheta)
  \quad\text{for every }\ell\geq0.
\end{equation}
Moreover, $U_m(\vartheta)>1/2$.
Let $a=\sin^2\vartheta$, $b=\sin^2((2m+1)\vartheta)$, and $c=\sin^2(2m\vartheta)$. The balancing relation also gives
\begin{equation}\label{eq:sharpness_critical_ratio}
  \frac{m(a+b-c)}{2(b-a)}
  =\frac14+\frac{1-\tan^2\vartheta}{8(2m-1)}.
\end{equation}
Since $a+c-b=-2\sin\vartheta\sin(2m\vartheta)\cos((2m+1)\vartheta)<0$, the value $\kappa=b-a-c+m^{-3}$ is positive. Set $X=(1+\kappa)^{-1}$ and $Y=\kappa/(1+\kappa)$, so that $X,Y>0$ and $X+Y=1$. Consider the two-point posterior $\mathbf{p}$ that assigns weights $X$ and $Y$ to the angles $\vartheta$ and $\pi/4$, respectively, and choose $r=m(a+\kappa/2)/(b-a)$.
For this posterior, direct testing and the depth-$m$ Grover searching instance have the same score. Equation~\eqref{eq:sharpness_balanced_depth} shows that this score equals $\max_{\ell\in\mathbb{Z}_{\geq0}}S(\ell,0,r\mid\mathbf{p})$, while eq.~\eqref{eq:sharpness_critical_ratio} gives
\begin{equation}\label{eq:sharpness_ratio_interval}
  \frac14<r<\frac14+\frac{1}{8(2m-1)}+\frac{1}{m^2}.
\end{equation}

Let $g_j$ denote the corresponding moments, set $\rho=X/(4m^3)$, and choose $\pi_0=\pi_L=(1-\rho)/2$ and $\pi_1=\pi_{L-1}=\rho/2$, with all other coefficients equal to zero. These coefficients satisfy eq.~\eqref{eq:positive_fourier_amplitude}. Moreover, $\Delta=g_1+g_{L-1}-g_L=X/m^3=4\rho$ and $\mathcal{B}=2g_1+2g_{L-1}-g_L-g_{L-2}\leq4$.
Let $\mathbf{s}$ be the scoring vector associated with this amplitude. Its norm, which is the posterior-predictive success probability defined in eq.~\eqref{eq:score_vector}, satisfies $\norm{\mathbf{s}}-g_L=2\rho\Delta-\rho^2\mathcal{B}=\rho^2(8-\mathcal{B})>0$. It therefore exceeds the complete envelope at the same depth and cost. Since the upper endpoint in eq.~\eqref{eq:sharpness_ratio_interval} approaches $1/4$ as $m\to\infty$, a counterexample exists below every prescribed $R>1/4$.
\end{proof}

The main consequence of Theorem~\ref{th:phase_score_bound} is that, for $0<r\leq1/4$, the greedy optimization may be restricted to $w=0$ without reducing the selected score. The same conclusion holds for every searching algorithm whose failure amplitude admits the decomposition in eq.~\eqref{eq:positive_fourier_amplitude} with all factors positive. Although the value $r=1/4$ is sharp for this full amplitude class, the limiting testing ratio for YLC fixed-point search under a particular posterior distribution may be higher. Moreover, any additional runtime required to implement $\mathbf{U}_k$ and $\mathbf{U}_k^\dagger$ increases the cost of a quantum iteration and therefore lowers the effective value of $r$. This dependency of the testing ratio with each searching instance is an additional advantage of the greedy model.

For the admissible searching instances, $w=0$ reduces the YLC fixed-point search to Grover search. Theorem~\ref{th:phase_score_bound} therefore gives, for $0<r\leq1/4$,
\begin{equation}\label{eq:ylc_no_score_advantage}
  \max_{\substack{m\in\mathbb{Z}_{\geq0}\\0\leq w<1}}
  S(m,w,r\mid\mathbf{p}^{(k)})
  =
  \max_{m\in\mathbb{Z}_{\geq0}}
  S(m,0,r\mid\mathbf{p}^{(k)}).
\end{equation}
We next quantify the possible score improvement as function of $r$ and $\mathbf{p}^{(k)}$.

\begin{theorem}\label{th:phase_score_gap}
For every posterior distribution $\mathbf{p}^{(k)}$, every $m\geq1$, and every $0\leq w<1$,
\begin{equation}\label{eq:phase_depth_score_gap}
  S(m,w,r\mid\mathbf{p}^{(k)})-S^*_{{Grover}}
  \leq
  S^*_{{Grover}}\frac{\delta_m(r)}{m+r},
\end{equation}
where
\begin{equation}\label{eq:phase_depth_excess}
  \delta_m(r)=
  \begin{cases}
    0, & 0<r\leq1/4,\\[2mm]
    \dfrac{4r-1}{2}
    \dfrac{\sinh(m\eta_r)}{\sinh((m+1)\eta_r)},
      & 1/4<r\leq1,
  \end{cases}
\end{equation}
and, for $r>1/4$, $\eta_r>0$ is defined by $\cosh\eta_r=\frac{4r}{4r-1}.$
\end{theorem}

\begin{proof} For $0<r\leq1/4$, eq.~\eqref{eq:phase_depth_score_gap} follows directly from Theorem~\ref{th:phase_score_bound}.

Take $r>1/4$ fixed and denote $a=4r-1>0$. Let $A(\theta)$ be the failure amplitude of the generalized search algorithm and let $\pi_j$ be its coefficients in eq.~\eqref{eq:positive_fourier_amplitude}, with $L=2m+1$ and $\pi_j=\pi_{L-j}$. Let $Z$ and $Z'$ be independent random variables satisfying $\Pr(Z=j)=\Pr(Z'=j)=\pi_j,$ and $J=|Z-Z'|$. Equation~\eqref{eq:phase_moment_reduction} gives $\norm{\mathbf{s}(m,w,\mathbf{p}^{(k)})}=2\mathbb{E}[g_J]\leq2S^*_{{Grover}}\mathbb{E}[c_r(J)]$.
Define $X=\min\{Z,L-Z\}$ and $X'=\min\{Z',L-Z'\}$. Since $L=2m+1$, both variables take values in $\{0,\ldots,m\}$. For every $x$ in the support of $X$, the symmetry $\pi_x=\pi_{L-x}$ gives
\begin{equation}\label{eq:phase_conditional_Z}
  \Pr(Z=z\mid X=x)=
  \begin{cases}
    1/2, & z=x,\\
    1/2, & z=L-x,\\
    0, & \text{otherwise}.
  \end{cases}
\end{equation}
Applying eq.~\eqref{eq:phase_conditional_Z} independently to $Z$ and $Z'$ gives
\begin{equation}\label{eq:phase_conditional_J}
  \Pr(J=j\mid X=x,X'=x')=
  \begin{cases}
    1/2, & j=|x-x'|,\\
    1/2, & j=L-x-x',\\
    0, & \text{otherwise}.
  \end{cases}
\end{equation}
Substituting the two possible values of $J$ from eq.~\eqref{eq:phase_conditional_J} into eq.~\eqref{eq:cr_definition}, with $a=4r-1$, gives
\begin{equation}\label{eq:phase_conditional_cr}
  2\mathbb{E}[c_r(J)\mid X=x,X'=x']
  =
  \begin{cases}
    m+r-x, & x=x',\\
    m+3r-\dfrac{1}{2}-\min(x,x'), & x\neq x'.
  \end{cases}
\end{equation}

Define the tail probabilities $u_j=\Pr(X\geq j)$, with $u_0=1$ and $u_{m+1}=0$. Since $\Pr(X=j)=u_j-u_{j+1}$ and $X,X'$ are independent and identically distributed, $\Pr(X\neq X')=1-\sum_{j=0}^{m}(u_j-u_{j+1})^2$ and $\mathbb{E}[\min(X,X')]=\sum_{j=1}^{m}\Pr(X\geq j,X'\geq j)=\sum_{j=1}^{m}u_j^2$.
Averaging eq.~\eqref{eq:phase_conditional_cr} therefore gives
\begin{equation}\label{eq:phase_tail_reduction}
  2\mathbb{E}[c_r(J)]
  =m+r+\frac{F_m(\mathbf{u})}{2},
\end{equation}
where
\begin{equation}\label{eq:phase_tail_quadratic}
  F_m(\mathbf{u})
  =a-a\sum_{j=0}^{m}(u_j-u_{j+1})^2
  -2\sum_{j=1}^{m}u_j^2.
\end{equation}

To bound this expression, we first maximize $F_m$ over all real values of $u_1,\ldots,u_m$ with the endpoints fixed. Since $a>0$, $F_m$ is a strictly concave quadratic function. For $j=1,\ldots,m$, its stationary equation is
\begin{equation}\label{eq:phase_tail_recurrence}
  a(u_{j-1}-2u_j+u_{j+1})-2u_j=0.
\end{equation}
Since $\cosh\eta_r=(a+1)/a$, eq.~\eqref{eq:phase_tail_recurrence} is equivalent to $u_{j-1}-2\cosh(\eta_r)u_j+u_{j+1}=0$.
Together with $u_0=1$ and $u_{m+1}=0$, its unique solution is
\begin{equation}\label{eq:phase_tail_optimizer}
  u_j^*
  =\frac{\sinh((m+1-j)\eta_r)}
         {\sinh((m+1)\eta_r)}.
\end{equation}
The sequence in eq.~\eqref{eq:phase_tail_optimizer} decreases from $u_0^*=1$ to $u_{m+1}^*=0$. It is therefore a valid sequence of tail probabilities, so the maximum of the relaxed problem is also attained in the admissible domain.

For completeness, the value of the maximum follows directly from the recurrence. Multiplying eq.~\eqref{eq:phase_tail_recurrence} by $u_j^*$, summing over $j=1,\ldots,m$, and using the two boundary values gives
\begin{equation*}
  a\sum_{j=0}^{m}(u_j^*-u_{j+1}^*)^2
  +2\sum_{j=1}^{m}(u_j^*)^2
  =a(1-u_1^*).
\end{equation*}
It follows from eq.~\eqref{eq:phase_tail_quadratic} that
\begin{equation*}
  F_m(\mathbf{u})
  \leq F_m(\mathbf{u}^*)
  =au_1^*
  =a\frac{\sinh(m\eta_r)}
          {\sinh((m+1)\eta_r)}
  =2\delta_m(r).
\end{equation*}
Combining this result with eqs.~\eqref{eq:phase_moment_reduction} and~\eqref{eq:phase_tail_reduction} gives
\begin{equation*}
  \norm{\mathbf{s}(m,w,\mathbf{p}^{(k)})}
  \leq
  S^*_{{Grover}}\left[m+r+\delta_m(r)\right].
\end{equation*}
Dividing by $m+r$ and using the score definition in eq.~\eqref{eq:score} proves eq.~\eqref{eq:phase_depth_score_gap}.
\end{proof}

From Theorem~\ref{th:phase_score_gap}, we conclude that the relative score excess of a YLC searching instance over Grover's best score satisfies $(S(m,w,r\mid\mathbf{p}^{(k)})-S^*_{{Grover}})/S^*_{{Grover}}\leq \frac{\delta_m(r)}{m+r}$, with $\delta_m(r)$ given by eq.~\eqref{eq:phase_depth_excess}. This upper bound decreases with $m$, so, regardless of the value of $r$, the possible improvement that can be obtained by adding the parameter $w$ is bounded more tightly at higher search depths. For $0<r\leq1/4$, this improvement is always 0. This proves that the upper bound on the advantage of fixed-point quantum search is larger at smaller search depths.

If we take into account that a generalized quantum amplitude amplification iteration $\mathbf{G}(\alpha, \beta)$ can be $\chi\geq1$ times more expensive to implement than the original $\mathbf{G}(\pi, \pi)$ iteration, we can also quantify the necessary upper limit on the cost ratio, $\chi_{c}$, for fixed-point search to provide an advantage. If we define the new score taking into account this correction as $S'(m,w,r,\chi\mid\mathbf{p}^{(k)})=\frac{m+r}{\chi m+r} S(m,w,r\mid\mathbf{p}^{(k)})$, from Theorem~\ref{th:phase_score_gap} we have that $S(m,w,r\mid\mathbf{p}^{(k)})\leq S^*_{{Grover}}(1+\frac{\delta_m(r)}{m+r})$, so $S'(m,w,r,\chi\mid\mathbf{p}^{(k)})>S^*_{{Grover}}$ implies
\begin{equation}
  \chi<\chi_{c}(r,m)
  :=1+\frac{\delta_m(r)}{m}.
\end{equation} 

Table~\ref{tab:phase_cost_score_bounds} gives the critical cost ratio and the corresponding relative score-gap bound for selected YLC depths $m$. The quantity $100(\chi_{c}-1)\%$ is the maximum percentage by which one YLC iteration may be slower than one standard Grover iteration while an improvement remains possible. The final column gives the upper bound on $100\Delta S_m/S^*_{{Grover}}$ from Theorem~\ref{th:phase_score_gap}.

\begin{table}[!h]
  \centering
  \caption{Depth-dependent YLC-iteration cost and relative score-gap bounds for a YLC searching instance of depth $m$.}
  \label{tab:phase_cost_score_bounds}
  \begin{tabular}{|c|c|c|c|c|}
    \hline
    $r$ & $m$ & $\chi_{c}(r,m)$ & Maximum slowdown & $\Delta S_m/S^*_{{Grover}}$ \\
    & & & $(\%)$ & upper bound $(\%)$ \\
    \hline
    $1.00$ & $1$  & $1.5625$ & $56.25$ & $28.13$ \\
    $1.00$ & $2$  & $1.3273$ & $32.73$ & $21.82$ \\
    $1.00$ & $5$  & $1.1354$ & $13.54$ & $11.28$ \\
    $1.00$ & $10$ & $1.0677$ & $6.77$  & $6.16$ \\
    \hline
    $0.75$ & $1$  & $1.3333$ & $33.33$ & $19.05$ \\
    $0.75$ & $2$  & $1.1875$ & $18.75$ & $13.64$ \\
    $0.75$ & $5$  & $1.0764$ & $7.64$  & $6.64$ \\
    $0.75$ & $10$ & $1.0382$ & $3.82$  & $3.55$ \\
    \hline
    $0.50$ & $1$  & $1.1250$ & $12.50$ & $8.33$ \\
    $0.50$ & $2$  & $1.0667$ & $6.67$  & $5.33$ \\
    $0.50$ & $5$  & $1.0268$ & $2.68$  & $2.44$ \\
    $0.50$ & $10$ & $1.0134$ & $1.34$  & $1.28$ \\
    \hline
    $0.25$ & $1$  & $1.0000$ & $0.00$ & $0.00$ \\
    $0.25$ & $2$  & $1.0000$ & $0.00$ & $0.00$ \\
    $0.25$ & $5$  & $1.0000$ & $0.00$ & $0.00$ \\
    $0.25$ & $10$ & $1.0000$ & $0.00$ & $0.00$ \\
    \hline
  \end{tabular}
\end{table}




\subsection{Classical-to-quantum score transition for large databases}\label{subsection:classical_quantum_transition}

In the large-database approximation, if the remaining database size after $k$ searching instances satisfies $N_k\sim N\gg 1$, the score can be approximated by
\begin{equation}
  S(m,w,r\mid\mathbf{p}^{(k)})
  \approx\frac{\displaystyle\int_0^1 P(\lambda,m,w)\,p^{(k)}(\lambda)\,{d}\lambda}{m+r},
\end{equation}
where $p_i^{(k)}\approx p^{(k)}(i/N_k)/N_k$, with $p^{(k)}(\lambda)$ the continuum posterior density corresponding to $\mathbf{p}^{(k)}$. In this limit, the change in the solution fraction caused by removing the tested candidates can be neglected.

Suppose that the first $k$ searching instances are classical and that $s\in\{0,\ldots,k\}$ return solutions. In the continuum approximation, let $o$ denote the outcome of the next classical searching instance, with $o=1$ if a solution is found and $o=0$ otherwise. The Bayesian update in eq.~\eqref{eq:bayes_update} then reduces to
\begin{equation}\label{eq:classical_continuum_update}
  p^{(k+1)}(\lambda\mid o)
  =\begin{cases}
    \dfrac{\lambda p^{(k)}(\lambda)}
    {\displaystyle\int_0^1 u p^{(k)}(u)\,{d}u},
      & \text{if }o=1,\\[6pt]
    \dfrac{(1-\lambda)p^{(k)}(\lambda)}
    {\displaystyle\int_0^1 (1-u)p^{(k)}(u)\,{d}u},
      & \text{if }o=0.
  \end{cases}
\end{equation}
Starting from the uniform prior $p^{(0)}(\lambda)=1$, repeated application of eq.~\eqref{eq:classical_continuum_update} for $s$ solutions and $k-s$ nonsolutions gives the normalized posterior density
\begin{equation}\label{eq:classical_beta_posterior}
  p^{(k)}(\lambda\mid s)
  =\frac{\lambda^s(1-\lambda)^{k-s}}
  {\mathcal{B}(s+1,k-s+1)},
\end{equation}
where $\mathcal{B}(a,b)=\int_0^1 u^{a-1}(1-u)^{b-1}\,{d}u$ is the beta function. Its mean is $(s+1)/(k+2)$, and the score of a classical searching instance is therefore
$S(0,0,r\mid\mathbf{p}^{(k)})=(s+1)/[r(k+2)]$. In this approximation, the posterior-predictive success probability is:
\begin{equation}\label{eq:beta-posterior-predictive}
\norm{\mathbf{s}(m,w,\mathbf{p}^{(k)})}=\int_0^1 P(\lambda,m,w)p^{(k)}(\lambda\mid s)\,{d}\lambda. 
\end{equation}

Let $k_{\mathrm{tr}}(r,s)$ be the smallest integer $k\geq s$ for which a quantum searching instance has a score at least as large as the classical score. This definition includes equality; a policy that retains the classical instance at a score tie instead uses a strict inequality when selecting its first quantum instance. We assume that the generalized Grover iterators require no extra execution time, $\chi=1$. From the score definition in eq.~\eqref{eq:score} and the classical score following eq.~\eqref{eq:classical_beta_posterior}, a quantum searching instance with $m\geq1$ reaches the classical score exactly when
\begin{equation*}
\norm{\mathbf{s}(m,w,\mathbf{p}^{(k)})}
\geq\frac{(m+r)(s+1)}{r(k+2)}.
\end{equation*}

For $m=1$, eq.~\eqref{eq:ylc_chebyshev} takes the form
\begin{equation*}
P(\lambda,1,w)
=\lambda+\frac{8\lambda(1-\lambda)}{1+3w^2}
-\frac{16\lambda^2(1-\lambda)}{(1+3w^2)^2}.
\end{equation*}

The posterior-predictive success probability given by eq.~\eqref{eq:beta-posterior-predictive} can be written in this case as
\begin{equation*}
\begin{split}
\norm{\mathbf{s}(1,w,\mathbf{p}^{(k)})}
={}&\frac{s+1}{k+2}
+\frac{8(s+1)(k-s+1)}{(k+2)(k+3)(1+3w^2)}\\
&-\frac{16(s+1)(s+2)(k-s+1)}{(k+2)(k+3)(k+4)(1+3w^2)^2}.
\end{split}
\end{equation*}

This equation is maximized at 
\begin{equation}\label{eq:depth_one_optimal_width}
w^2=
\begin{cases}
\dfrac{4s+4-k}{3(k+4)}, & k<4s+4,\\[6pt]
0, & k\geq4s+4.
\end{cases}
\end{equation}

Equating the optimized depth-one score with the classical score gives the critical testing ratio
\begin{equation}\label{eq:depth_one_critical_ratio}
r_{{crit}}^{(1)}(k,s)
=\begin{cases}
\dfrac{(s+2)(k+3)}{(k-s+1)(k+4)},
& s\leq k<4s+4,\\[8pt]
\dfrac{(k+3)(k+4)}{8(k-s+1)(k-2s)},
& k\geq4s+4.
\end{cases}
\end{equation}

For $m\geq2$, the probability bound gives $S(m,w,r\mid\mathbf{p}^{(k)})\leq1/(m+r)\leq1/(2+r)$. Moreover, a finite-depth searching instance has posterior-predictive success probability strictly less than one under the Beta posterior. A necessary condition for such an instance to reach the classical score is therefore $r(k-s+1)>2(s+1)$. Consequently, the critical testing ratio for any depth $m\geq2$ satisfies
\begin{equation}\label{eq:depth_two_critical_ratio}
  r_{\mathrm{crit}}^{(\geq2)}(k,s)>\frac{2(s+1)}{k-s+1}.
\end{equation}

The two branches of eq.~\eqref{eq:depth_one_critical_ratio} satisfy
\begin{equation*}
  r_{\mathrm{crit}}^{(1)}(k,s)
  <\max\left\{\frac14,\frac{2(s+1)}{k-s+1}\right\},
  \qquad k\geq s.
\end{equation*}
For $r>1/4$, if the optimized depth-one score is still below the classical score, then $r<r_{\mathrm{crit}}^{(1)}(k,s)$ and this inequality forces $r<2(s+1)/(k-s+1)$. Equation~\eqref{eq:depth_two_critical_ratio} then excludes every larger depth. Thus, no larger depth reaches the classical score before depth one, and the first quantum score crossing is the depth-one crossing in the range $1/4<r\leq1$.

By solving eq.\eqref{eq:depth_one_critical_ratio} for $k$, we obtain on two branches the value of $k_{tr}(s,r)$ as:
\begin{equation}\label{eq:transition_positive_width_branch}
  k_{tr}^+(s,r)=\left\lceil s-1+
\frac{s+2-r(s+3)+\sqrt{\left[s+2-r(s+3)\right]^2+4r(s+2)^2}}{2r}\right\rceil,
\end{equation}
when $\dfrac{4s+7}{4(3s+5)}\leq r\leq1$, respectively
\begin{equation}\label{eq:transition_zero_width_branch}
\begin{aligned}
k_{tr}^-(s,r)
={}&\Biggl\lceil
s-1+\frac{8r(s+1)+2s+5}{2(8r-1)}
\\
&\qquad
+\frac{
\sqrt{\left[8r(s+1)+2s+5\right]^2
+4(8r-1)(s+2)(s+3)}
}{2(8r-1)}
\Biggr\rceil .
\end{aligned}
\end{equation}
when $\dfrac14<r<\dfrac{4s+7}{4(3s+5)}$. The superscripts distinguish the two analytic branches before integer rounding: eq.~\eqref{eq:transition_positive_width_branch} can also give $k_{\mathrm{tr}}=4s+4$, where the optimized width is zero. Both formulas include score ties. For a strict depth-one crossing, the ceiling of the real root is replaced by its floor plus one; the result differs only when the root is an integer. In the particular case where no solutions have been found, $s=0$, the transition points are summarized in Table~\ref{tab:transition-s0}.

\begin{table}[htbp]
\centering
\caption{Classical-to-quantum transition points for $s=0$.}
\label{tab:transition-s0}
\setlength{\tabcolsep}{5pt}
\begin{tabular}{|c|*{15}{c}|}
\hline
$r$
& $1/4$ && $15/56$ && $3/10$
&& $7/20$ && $3/7$ && $5/9$
&& $4/5$ && $1$ \\
\hline
$k_{\mathrm{tr}}(r,0)$
& & $7$ & $\big|$ & $6$ & $\big|$ & $5$
& $\big|$ & $4$ & $\big|$ & $3$ & $\big|$ & $2$
& $\big|$ & $1$ & \\
\hline
\end{tabular}
\end{table}

In the case $r\leq 1/4$, Theorem~\ref{th:phase_score_bound} reduces the complexity of finding the transition condition only to the case of Grover searching instances. Then, the transition point $k_{tr}(r,s)$ is the smallest integer $k\geq s$ such that
\begin{equation}\label{eq:classical_quantum_transition_condition}
\max_{m\in\mathbb{Z}_{\geq1}}
\frac{1}{m+r}\int_0^1 \sin^2\!\left((2m+1)\arcsin\sqrt{\lambda}\right) p^{(k)}(\lambda\mid s){d}\lambda
\geq\frac{s+1}{r(k+2)},
\end{equation}
where $p^{(k)}(\lambda\mid s)$ is the continuous posterior in eq.~\eqref{eq:classical_beta_posterior}. We determine its asymptotic behavior for $r\to 0^+$ in Theorem~\ref{th:classical_quantum_transition}, as $k_{{tr}}(r,s)=\frac{C_s^2}{r^2}-\frac{2B_s}{r}+\mathcal{O}_s(1)$, where $C_s=\frac{(s+1)x_s}{2q_s(x_s)}$, $B_s=\frac{s+1}{2q_s(x_s)}$, and
\begin{equation}\label{eq:limiting_predictive_success}
  q_s(x)=\frac{1}{s!}\int_0^\infty u^s e^{-u}\sin^2(x\sqrt{u})\,\text{d}u
\end{equation}
 with $x_s$ the unique positive global maximum of $q_s(x)/x$.

\begin{theorem}\label{th:classical_quantum_transition}
Within the continuous posterior approximation in eq.~\eqref{eq:classical_beta_posterior}, for fixed $s$ and $r\to0^+$,
\begin{equation}\label{eq:classical_quantum_transition}
k_{{tr}}(r,s)
=\frac{C_s^2}{r^2}
-\frac{2B_s}{r}
+\mathcal{O}_s(1).
\end{equation}
\end{theorem}
\begin{proof}
Using the variable change $u=(k+2)\lambda$ in the left hand side of eq.~\eqref{eq:classical_quantum_transition_condition} gives
\begin{equation*}
\begin{split}
\norm{\mathbf{s}(m,0,\mathbf{p}^{(k)})}
={}&\frac{1}{(k+2)^{s+1}B(s+1,k-s+1)}
\int_0^{k+2}u^s\left(1-\frac{u}{k+2}\right)^{k-s}\\
&\times\sin^2\!\left((2m+1)\arcsin\sqrt{\frac{u}{k+2}}\right)\,{d}u.
\end{split}
\end{equation*}

For fixed $s$, we have that $(1-\frac{u}{k+2})^{k-s}=e^{-u}(1+\mathcal{O}_s(\frac{1}{k+2}))$. Since $P(\lambda,m,0)\leq\min\{1,(2m+1)^2\lambda\}$, a strict improvement over the classical score requires $\frac{1}{4r}-1<m<r\frac{k-s+1}{s+1}$. Then, we can write $(2m+1)\arcsin\sqrt{\frac{u}{k+2}}=\frac{2m+1}{\sqrt{k+2}}\sqrt{u}+\mathcal{O}(\frac{1}{k+2})$ and express the posterior-predictive success probability as:
\begin{equation}\label{eq:predictive_success_asymptotic}
\norm{\mathbf{s}(m,0,\mathbf{p}^{(k)})}
=q_s\!\left(\frac{2m+1}{\sqrt{k+2}}\right)
+\mathcal{O}_s\!\left((k+2)^{-1}\right).
\end{equation}

Equality of the classical and Grover scores gives
\begin{equation*}
r=\frac{(s+1)m}
{(k+2)\norm{\mathbf{s}(m,0,\mathbf{p}^{(k)})}-(s+1)},
\end{equation*}
and, defining $x=(2m+1)/\sqrt{k+2}$ and substituting eq.~\eqref{eq:predictive_success_asymptotic}, the transition value of $r$ from a classical to a quantum searching instance at fixed $k$ is
\begin{equation*}
r_{cr}(k,s)
=\min_{m\in\mathbb{Z}_{\geq1}}
\left\{
\frac{(s+1)x}{2q_s(x)\sqrt{k+2}}
-\frac{s+1}{2q_s(x)(k+2)}
+\mathcal{O}_s\!\left((k+2)^{-3/2}\right)
\right\}.
\end{equation*}
As the leading term is minimized at $x=x_s$,
\begin{equation*}
r_{cr}(k,s)
=\frac{C_s}{\sqrt{k+2}}-\frac{B_s}{k+2}
+\mathcal{O}_s\!\left((k+2)^{-3/2}\right).
\end{equation*}
At the first crossing, inversion at the two adjacent integer values of $k$ gives
\begin{equation*}
\sqrt{k_{{tr}}(r,s)+2}
=\frac{C_s}{r}-\frac{B_s}{C_s}+\mathcal{O}_s(r).
\end{equation*}
Squaring proves eq.~\eqref{eq:classical_quantum_transition}.
\end{proof}

For the case of no solutions found, $s=0$, eq.~\eqref{eq:limiting_predictive_success} reduces to
$q_0(x)=xe^{-x^2}\int_0^x e^{t^2}\,{d}t$. The ratio $q_0(x)/x$ has a unique positive maximum at $x_0$, characterized by
$2x_0e^{-x_0^2}\int_0^{x_0}e^{t^2}\,{d}t=1$. Therefore, $q_0(x_0)=1/2$, $C_0=x_0$, and $B_0=1$, and eq.~\eqref{eq:classical_quantum_transition} becomes
\begin{equation}\label{eq:classical_quantum_transition_s0}
k_{{tr}}(r,0)+2
=\frac{x_0^2}{r^2}-\frac{2}{r}+\mathcal{O}(1)
\end{equation}
For each value of $s$ in Table~\ref{tab:classical_quantum_coefficients}, numerical maximization gives a unique global maximum of $q_s(x)/x$.
\begin{table}[!ht]
\centering
\caption{Numerical maximizers and asymptotic coefficients in eq.~\eqref{eq:classical_quantum_transition}.}
\label{tab:classical_quantum_coefficients}
\begin{tabular}{cccc}
\toprule
$s$ & $x_s$ & $C_s^2$ & $2B_s$ \\
\midrule
$0$ & $0.924139$ & $0.854033$ & $2.000000$ \\
$1$ & $0.722993$ & $1.332033$ & $3.192663$ \\
$2$ & $0.614073$ & $1.809024$ & $4.380586$ \\
$3$ & $0.543159$ & $2.285651$ & $5.566833$ \\
$4$ & $0.492265$ & $2.762105$ & $6.752294$ \\
\bottomrule
\end{tabular}
\end{table}

\section{Numerical simulation analysis}
In this section we present the numerical setup used for simulating the performance of the adaptive bayesian search algorithm with the fixed-point quantum search generalization and compare it with the performance of the one using Grover's unmodified search. We used two different numerical approaches for simulating artificial data, both based on the version of the algorithm starting with the uniform posterior distribution. The first approach bulilds a decision tree, analysing all possible outcomes after 20 measurements, while the second approach is a Monte-Carlo simulation for analysing the performance of the full searching algorithms presented. The numerical results are also confirming the results of the theoretical analysis in Section~\ref{sec:theory}.

\subsection{Numerical setup and simulation methods}\label{subsec:numerical_setup}

We performed the calculations on classical hardware using the analytic success probabilities in the ideal oracle model. We did not simulate gate-level circuits, quantum state vectors, device noise, or quantum hardware. Each search instance generated a Bernoulli observation conditioned on the true remaining solution fraction, with the corresponding success probability as its parameter. Given the permutation symmetry assumed for the unstructured database, this simulation implements the procedure of Section \ref{sec:algorithm}.

A reference implementation ran on the CPU with NumPy and SciPy using double precision arithmetic. The YLC Chebyshev ratio was evaluated in the log domain to prevent overflow, with endpoints handled explicitly, with the exact Grover limit at $w=0$. Posteriors were normalized by log-sum-exp. For width optimization, the reference code searched $q=w^2$ over both linear and logarithmic coordinates with DIRECT, then applied bounded scalar refinement. Endpoints were included in the search, and the boundary derivative was checked.

We complemented the reference CPU implementation with a CUDA-accelerated implementation, in order to make large simulation sizes tractable. We used a single NVIDIA H100 GPU to perform three large-scale simulation campaigns. Each of the three H100 studies started from a width grid of 17 points, with a maximum grid budget of 257 and a depth cap of 4096. The grid budget was distinct from the continuous refinement stage. The parameter sweep analysis used the same width bounds and score and width tolerances as the tree studies described below. While scanning depths, the bound \(1/(m+r)\) could rule out deeper actions once they could no longer beat the current score. Strict optimizer mode was on for the sweep and the continuum tree, and off for the finite tree. Later verification confirmed optimizer convergence for all retained finite-tree and continuum beliefs.

CUDA acceleration relied on CuPy with fused double-precision kernels. When posterior states shared a support, a single table of success probabilities over support, depth, and width was computed, and dense matrix multiplication produced the posterior expectations for every state in the batch. The continuum calculation followed the same route on its fixed quadrature nodes. For the sweep, which handles arbitrary supports, the computation proceeded in depth blocks using fused multiplication and reduction kernels. Probability tables were reused where compatible, and large arrays were split into chunks to bound memory usage. These procedures were done in order to accelerate the computation for the datapoints. In all CUDA-accelerated studies, continuous width refinement was limited to the 16 depths with the largest grid scores, while the remaining depths kept their grid values. On the validation workloads this choice agreed with refining every depth, although the comparison is not a proof of global optimality. For concentrated posteriors, GPU reductions used a contiguous support band selected by a weight threshold of \(10^{-22}\) times the row maximum. This restriction applied only to score evaluation and was independent of posterior compression: the sweep used an atom budget of 4097, so the full count support could be retained.

Within all experiments, Bayesian updates, belief merging, exact recursions, and the small workloads stayed on the CPU, as it is a more appropriate medium for such computations. For action selection, large batches moved between host and device memory. Cases falling outside the accelerated calculation's validity conditions fell back to the scalar implementation. Random streams were keyed to trial identities, so worker scheduling had no effect on their assignment. On smaller cases that exercised nontrivial policy choices, we have cross-validated results from the GPU with those of the reference CPU implementation, by comparing probabilities, selected actions, and aggregate quantities. Bit-identical results were not required, as the calculation paths are different between the implementations.

Each H100 study used a single NVIDIA H100 80 GB HBM3 GPU. The sweep ran with 12 CPU worker processes, the finite tree with two, and the continuum tree with one. All three runs were on Linux x86-64 with kernel 6.18.9-1.el10.elrepo.x86\_64 and glibc 2.39, under CPython 3.13.11, NumPy 2.4.6, SciPy 1.17.1, CuPy 14.1.1, and CUDA runtime 13.2. These platform details pertain to the H100 source datasets with the posterior calculations described in the transition subsection constituting a separate set of computations.

\subsection{Finite and continuum-tree analysis}\label{subsec:tree_analysis}

For validating the results of Theorems~\ref{th:phase_score_bound} and~\ref{th:phase_score_gap} in real full searching algorithms, we build a decision tree of all possible outcomes after 20 measurements. We first construct the tree for a finite database of size $N=16384$ starting from the uniform prior $\mathbf{p}^{(0)}$ with $p^{(0)}_i=\frac{1}{N+1}$, $\forall i\in\{0, 1, \ldots, N\}$. Then, we choose the optimal parameters $(m^*, w^*)$ that maximize the score for this posterior distribution for the given $r$ value. After that, we update the posterior distribution for both possible outcomes of the measurement and repeat the process until we reach the maximum depth of 20. The tree is constructed for each testing ratio $r\in\{0.20,0.25,\ldots,1\}$. 

\begin{figure}[!h]
  \centering
  \includegraphics[width=\textwidth]{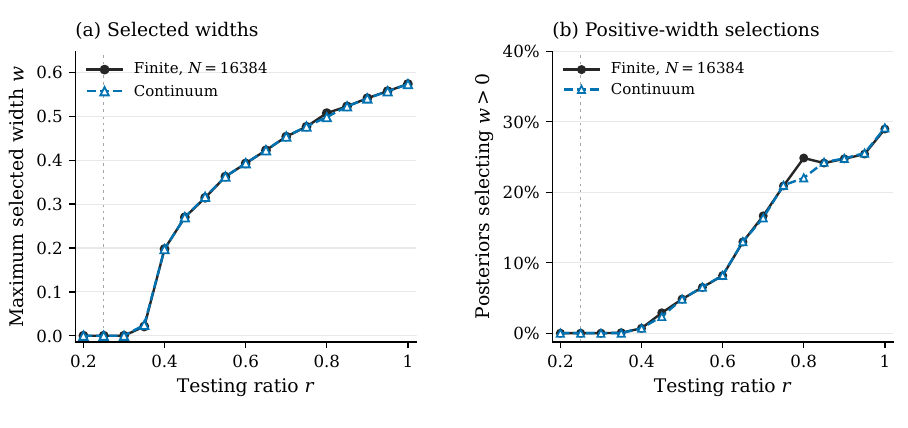}
  \caption{Selected widths in the finite $N=16384$ and continuum trees through level $20$, as a function of $r$: (a) the maximum selected width at each ratio; (b) the percentage of retained posterior records selecting $w>0$. Dashed vertical lines mark $r=1/4$.}
  \label{fig:finite_continuum_reachable_tree_widths}
\end{figure}

Additionally, we create another decision tree for the continuum limit $N\to \infty$, where the posterior distributions used are now continuous and are updated neglecting the finite-population removal step. In the continuous case, we describe the uncertainty in the solution fraction $\lambda\in[0,1]$ using a posterior density $p^{(k)}(\lambda)$, starting from the uniform prior $p^{(0)}(\lambda)=1$. We replace the sums used to calculate the scores with integrals, evaluated numerically using a $2048$-point Gauss--Legendre quadrature rule, and select the parameters $(m^*,w^*)$ that maximize the score. For each possible outcome, the posterior weights are multiplied by the corresponding likelihood, $P(\lambda,m^*,w^*)$ for success or $1-P(\lambda,m^*,w^*)$ for failure, and then normalized. The quadrature nodes remain fixed, since the change in the solution fraction caused by removing one candidate vanishes as $N\to\infty$ at a fixed tree level. We construct the continuum tree up to level $20$ and use this common depth for comparisons with the finite tree.

For both trees, we determine $(m^*,w^*)$ by examining successive integer values of $m$ and optimizing the width over $0\leq w\leq0.999$ at each depth. Since the cost $m+r$ is fixed for a given $m$, the width optimization maximizes the posterior-predictive success probability. We use the variable $q=w^2$, with a GPU-evaluated coarse width search followed by bounded scalar refinement, and an additional search in $\log q$ when needed to resolve narrow maxima near $w=0$. The configured width and objective-value tolerances are both $10^{-10}$, and the endpoints remain explicit candidates. Starting from the classical instance $m=0$, we increase $m$ until the upper bound $1/(m+r)$ cannot exceed the best score already found by more than $10^{-10}$, with a numerical safety limit of $m=4096$. The selected parameters correspond to the largest computed score with exact score ties are resolved in favour of the smaller $m$ and then the smaller $w$.

For each distinct posterior state, we store the testing ratio $r$, tree level $k$, selected parameters $(m^*,w^*)$, predicted success probability, and corresponding score, together with the optimal Grover score evaluated on the same posterior. The parent-child relationships are also recorded in order to analyze the transition from classical to quantum searching instance selection. 

\begin{figure}[!h]
  \centering
  \includegraphics[width=\textwidth]{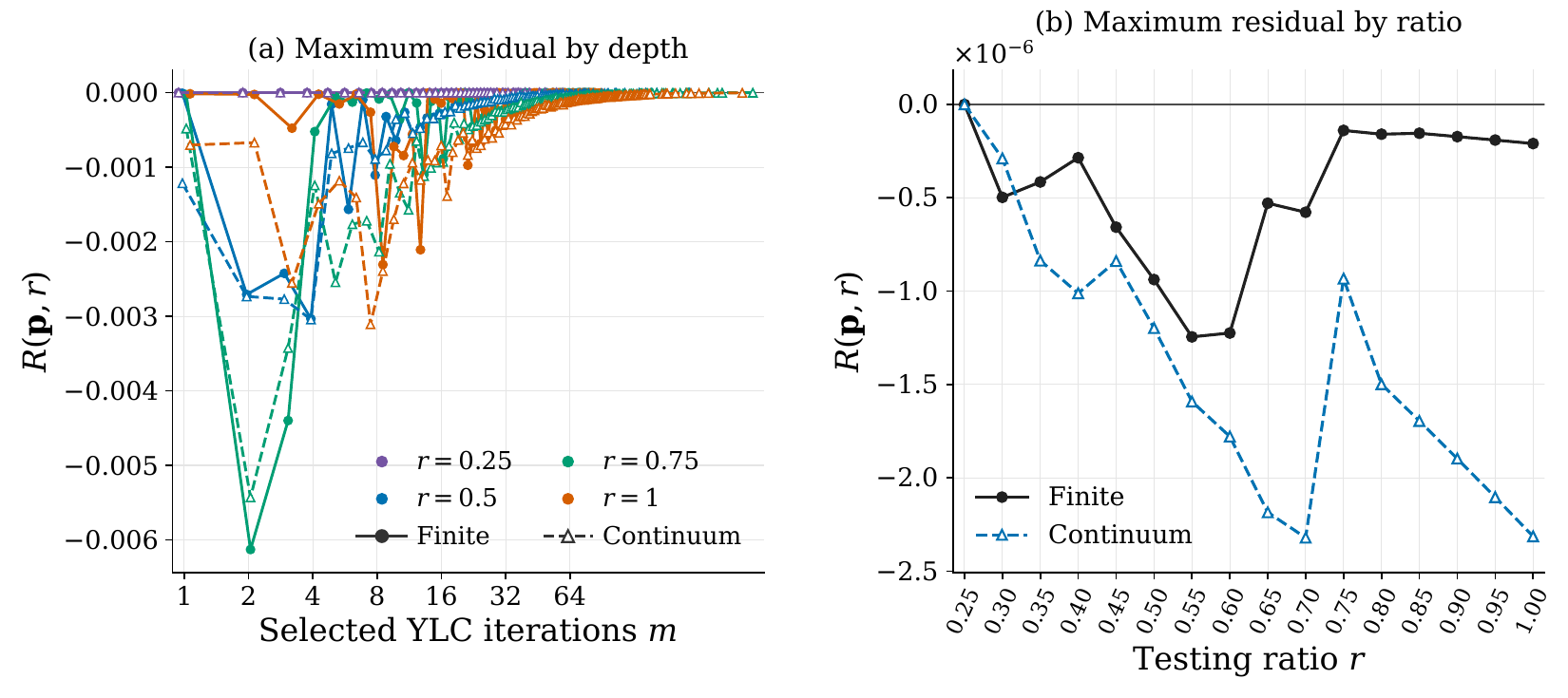}
  \caption{Residual of the maximum gap between the selected YLC and the optimal Grover score for the same posterior distribution obtained from the $N=16384$ tree and the continuum tree through level $20$: (a) as a function of the YLC depth $m$ for fixed testing ratio $r\in\{0.25, 0.5, 0.75, 1\}$; (b) as a function of the testing ratio $r$ over all selected depths.}
  \label{fig:posterior_specific_tree_residuals}
\end{figure}

We analyze the dependence of the selected parameter $w$ on the testing ratio $r$ in order to numerically check Theorem~\ref{th:phase_score_bound}. In Fig.~\ref{fig:finite_continuum_reachable_tree_widths}, we show the maximum selected width $w$ over all obtained posterior distributions at levels $0$--$20$ of the finite and continuum trees, as well as the percentage of these posteriors that selected $w>0$. The maximum selected width decreases as the testing ratio decreases, while the frequency of selecting a positive width follows this overall trend with some nonmonotonic variations. At the sampled ratios $r\leq0.30$, no posterior selected a positive width in either tree. At $r=0.35$, $1400$ of the $2097151$ finite-tree posterior records selected a positive width ($\approx0.0668\%$), with a maximum of $w\approx0.0211$, while $140$ records selected a positive width in the continuum tree ($\approx0.00668\%$), with a maximum of $w\approx0.0247$. These observations are consistent with the absence of a positive-width score advantage for $r\leq1/4$, as stated in Theorem~\ref{th:phase_score_bound}.

For testing the maximum gap between the selected YLC and the optimal Grover for the same posterior distribution, described in Theorem~\ref{th:phase_score_gap}, we calculate the residual $R(\mathbf p,r)$ for each retained state in the finite and continuum trees through level $20$. The residual is calculated as the difference between the numerical gap between the optimal YLC and Grover scores and the theoretical one predicted by Theorem~\ref{th:phase_score_gap}. In Fig.~\ref{fig:posterior_specific_tree_residuals}, we show the maximum residual over all posteriors obtained by the two trees, both as a function of the selected YLC depth $m$ at fixed testing ratios and as a function of the testing ratio $r$ over all selected depths. The residuals remain at numerical zero at $r=0.25$ and stay nonpositive across the sampled depths and ratios. Across all selected depths, the largest residual is $5.3\times10^{-15}$, and no state exceeds the posterior-specific bound within the numerical allowance.

These results confirm the theoretical results obtained in Theorems~\ref{th:phase_score_bound} and~\ref{th:phase_score_gap}, with no violation of the theoretical bounds being observed in the generated data.

\subsection{Paired complete-search benchmark}\label{subsec:paired_benchmark}

For characterizing and comparing the performance of the two greedy algorithms, we simulate complete searches on a database of size $N=4096$, starting from the uniform prior on $M\in\{0,\ldots,N\}$. We consider the testing ratios $r\in\{0.25,0.50,0.75,1\}$ and the true solution counts $M\in\{1,\ldots,32,64,256,1024,4096\}$. For each $(M,r)$, we use $1024$ paired trials, with the same labelled random streams for Grover and YLC, also reused across the testing ratios. The resulting data contain $294\,912$ trajectories. We measure the runtime by $T=Q_{\mathrm q}+rQ_{\mathrm c}$, so that each searching instance contributes its quantum depth $m$ and the verification cost $r$.

\begin{figure}[!h]
  \centering
  \includegraphics[width=\textwidth]{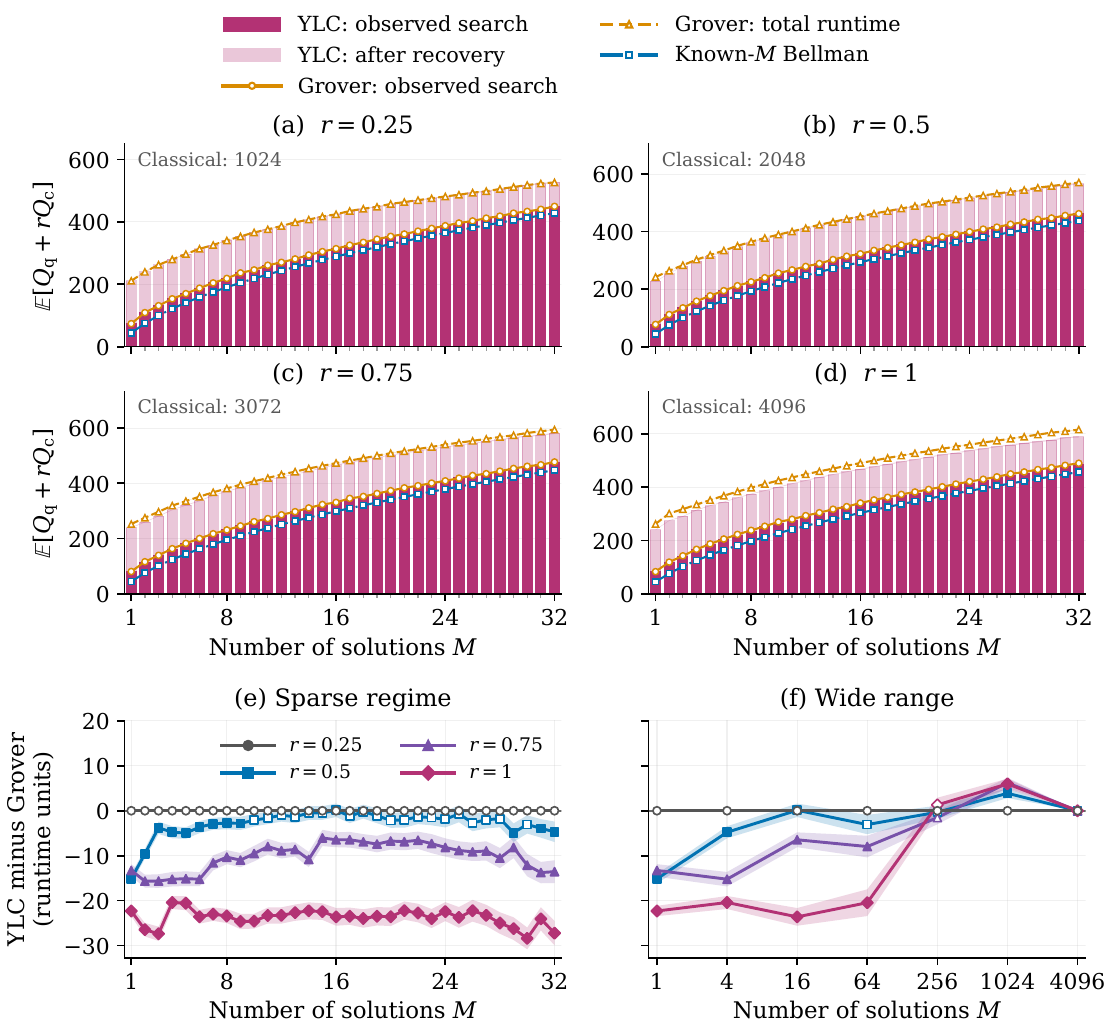}
  \caption{Complete-search performance at $N=4096,\varepsilon=0.10$: (a)-(d) YLC observed search expenditure (saturated bars) and expenditure after full recovery (pale bars), with Grover search expenditure (solid) and total runtime (dashed); (e)-(f) paired YLC-minus-Grover total-runtime differences over sparse and wide counts. The Bellman curve gives the known-$M$ discovery reference, and annotations give the classical cost $rN$. All costs use $Q_{\mathrm q}+rQ_{\mathrm c}$ and all $1024$ trials. Shaded bands give pointwise $95\%$ confidence intervals; filled markers in (e)-(f) indicate Holm-resolved differences.}
  \label{fig:unknown_m_complete_comparison}
\end{figure}

For each searching instance, we record the selected parameters $(m_j,w_j)$, the outcome, and the posterior expected number of remaining solutions $e_j$. The source searches were recorded to $\varepsilon=0.10$ and compares it with $0.25$, $0.50$, and $1$. For each tolerance, we retain the history through the first instance satisfying $e_j\leq\varepsilon$, including its full cost. We also verified the recorded histories and reconstructed resource totals against the archived benchmark. At $\varepsilon=0.10$, $2114$ trajectories stop before recovering every solution, representing less than $0.72\%$ of the recorded trajectories.

In Fig.~\ref{fig:unknown_m_complete_comparison}(a)-(d), we show the mean runtime and its expenditure before and after finding all solutions for $M=1,\ldots,32$, decomposing the total runtime of the algorithm into finding all solutions and confirming there are no solutions in the database. The saturated part of each YLC bar includes the cost through the instance finding the final solution, or the entire runtime if the search stops earlier. We call this component observed search expenditure. The pale part contains only expenditure after full recovery, and the two parts sum to the total runtime on every trial. The Grover lines show the observed search expenditure and total runtime. We also include the known-$M$ Bellman result as a discovery-cost reference, while the exhaustive classical cost $rN$ is given in each panel. All means use the same $1024$ trials per algorithm.

\begin{figure}[!h]
  \centering
  \includegraphics[width=\textwidth]{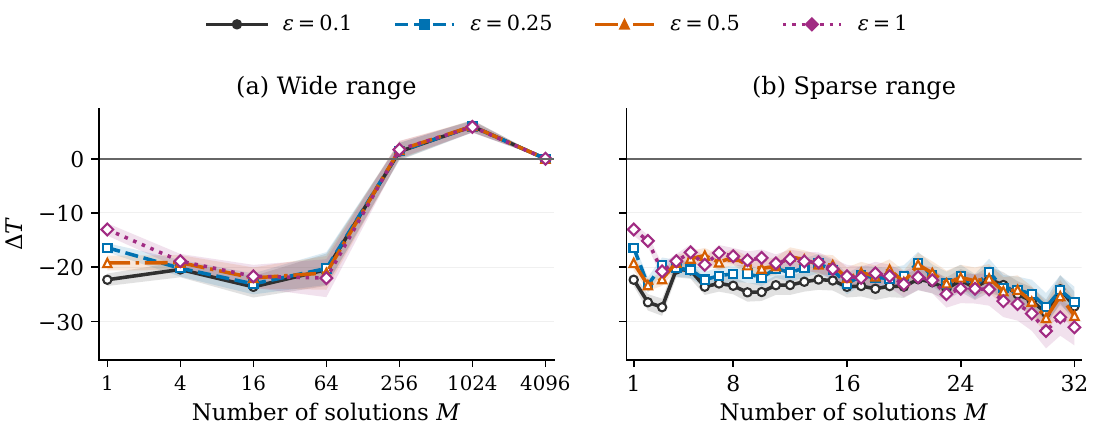}
  \caption{Paired runtime difference $\Delta T=T_{\mathrm{YLC}}-T_{\mathrm{Grover}}$ at $N=4096,r=1$ for $\varepsilon\in\{0.10,0.25,0.50,1\}$: (a) wide solution counts on a logarithmic axis; (b) every $M=1,\ldots,32$ on a linear axis. Negative values favor YLC. Shaded bands give pointwise $95\%$ paired confidence intervals, including all $1024$ trials.}
  \label{fig:stopping_sensitivity}
\end{figure}

Fig.~\ref{fig:unknown_m_complete_comparison}~(e)-(f) show the paired difference $\Delta T=T_{\mathrm{YLC}}-T_{\mathrm{Grover}}$, with negative values indicating a YLC advantage. For $r=1/4$, all $36\,864$ pairs have identical actions, outcomes, and posterior means, with $w=0$ throughout, reconfirming the result of Theorem~\ref{th:phase_score_bound}. For $r=0.50$, $0.75$, and $1$, YLC has a significantly lower mean runtime at $12$, $33$, and $33$ of the $36$ tested counts, respectively. Their runtime become approximately equal, with a difference of less than one Grover iteration time at $M=256$, and a small advantage for Grover in the range $256-4096$. These comparisons use Holm correction over the $144$ baseline comparisons. At $M=16,r=1$, for example, the mean runtimes are $490.57$ for Grover and $466.94$ for YLC.

\begin{figure}[!h]
  \centering
  \includegraphics[width=\textwidth]{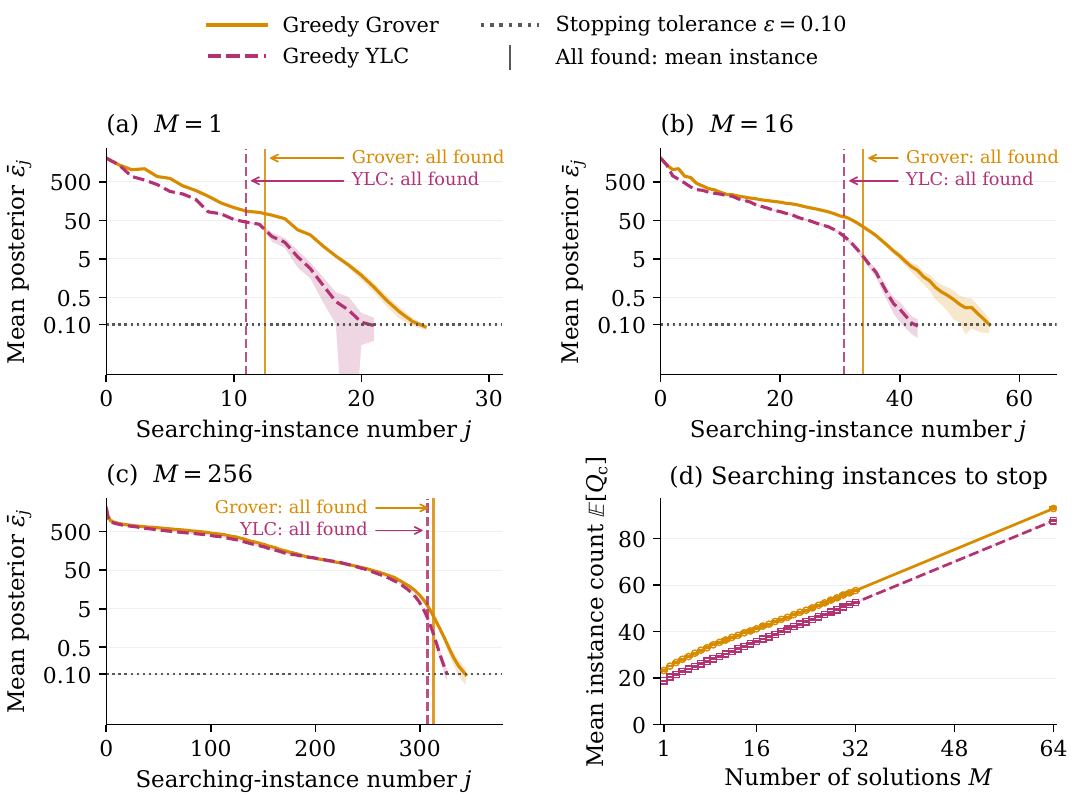}
  \caption{Posterior evolution at $N=4096,r=1$: (a)-(c) average remaining-solution means for $M=1,16,256$, with terminal values held after stopping and curves ending at their first mean crossing of $0.10$; (d) mean searching-instance count for $M=1,\ldots,32,64$ on linear axes. Shaded bands and capped error bars give pointwise $95\%$ confidence intervals.}
  \label{fig:posterior_instance_histories}
\end{figure}

In Fig.~\ref{fig:stopping_sensitivity}, we analyze how the stopping tolerance changes the runtime difference at $r=1$. Relaxing the tolerance decreases the mean runtime of both algorithms, but their relative performance need not change in the same direction. For $M=16$, YLC saves between $21.68$ and $23.62$ runtime units across the four tolerances. For $M=256$, the difference is smaller and favors Grover: at $\varepsilon=1$, it is $1.69$, with a pointwise $95\%$ interval of $[0.10,3.28]$, indicating a small advantage in this exploratory comparison. The shorter runtime must also be considered together with the probability of missing solutions. At $M=256,r=1$, increasing $\varepsilon$ from $0.10$ to $1$ reduces the mean runtime by $4.70\%$ for Grover and $4.67\%$ for YLC, while the numbers of incomplete searches increase from $45$ to $574$ and from $51$ to $575$, respectively, out of $1024$ trials.

In Fig.~\ref{fig:posterior_instance_histories}(a)-(c), we show the evolution of the average posterior expected number of remaining solutions, overe all $1024$ trials, as a function of the searching-instance number for $M=1,16,256$ at $r=1$, stopping when the stopping condition $\varepsilon=0.10$ is achieved. The vertical lines mark the average instance at which all solutions are found, calculated only on jointly successful pairs.

Panel~(d) shows the mean number of searching instances required to reach the stopping condition, for the measured counts $M=1,\ldots,32,64$. Each instance includes one verification, so this number is $Q_{\mathrm c}$. At $M=1,r=1$, Grover uses a mean of $23.23$ instances and YLC uses $18.88$, with mean runtimes of $262.10$ and $239.79$. The mean stopping instance is calculated from individual trials and need not coincide with the crossing of an averaged posterior curve. Instance counts also do not determine the runtime alone, since the selected quantum depths contribute to the cost.

Confidence intervals use Student-$t$ estimates, with trial pairing retained for runtime differences. These results describe complete-search performance at the tested $N$.

\subsection{Classical-to-quantum transition}\label{subsec:numerical_transition}

For checking the classical-to-quantum transition described in Section~\ref{subsection:classical_quantum_transition}, we analyze the first quantum searching instance selected by the algorithm. From both trees, we retain the states for which all preceding instances were classical. We analyze the discrete transition conditions for $r>1/4$ and the asymptotic expression for smaller ratios separately.

\begin{figure}[!h]
  \centering
  \includegraphics[width=\textwidth]{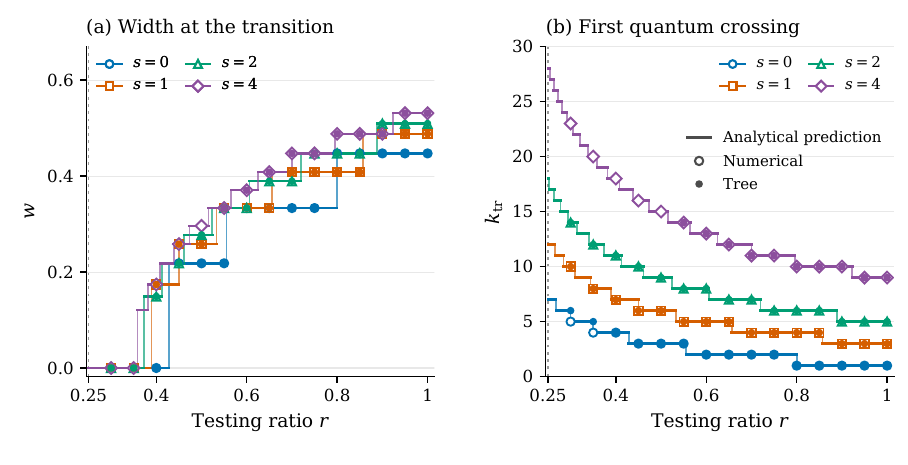}
  \caption{Classical-to-quantum transition for $s\in\{0,1,2,4\}$ and $r>1/4$: (a) selected width; (b) transition count. Lines show analytical predictions, open symbols show direct posterior calculations, and filled circles show continuum-tree selections, including delays at score ties. Colors identify $s$.}
  \label{fig:classical_quantum_discrete_transition}
\end{figure}

In Fig.~\ref{fig:classical_quantum_discrete_transition}, we compare the continuum-tree selections with direct calculations on the same family of Beta posteriors, which also allow us to examine transitions beyond level $14$. We evaluate the success probabilities by Gauss-Jacobi quadrature and optimize the widths independently of the theoretical predictions. Panel~(a) shows the selected width together with eq.~\eqref{eq:depth_one_optimal_width}, while panel~(b) shows the transition count together with eqs.~\eqref{eq:transition_positive_width_branch} and~\eqref{eq:transition_zero_width_branch}. All $558$ directly calculated transitions for $s=0,\ldots,4$ agree with the predictions, with $m=1$ giving the first quantum crossing. The maximum width differences are $5.50\times10^{-15}$ for the direct calculations and $1.26\times10^{-8}$ for the saved continuum-tree selections.

The transition formulas include equality between the quantum and classical scores. At an exact tie, the saved optimizer can retain a classical instance, explaining the one-instance delays observed at $s=0$, $r=0.30$ and $r=0.35$. The width prediction is therefore evaluated at the actual selected $k$. The staircase structure follows from the integer number of preceding classical instances.

\begin{figure}[tbp]
  \centering
  \includegraphics[width=\textwidth]{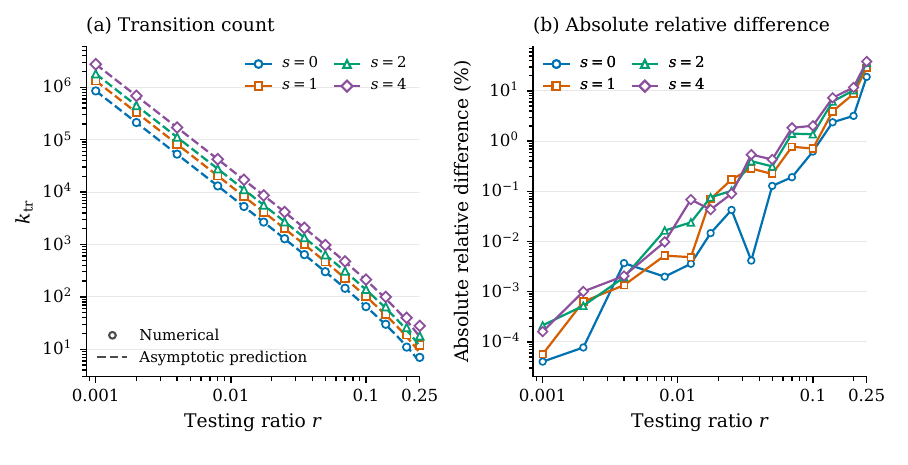}
  \caption{Asymptotic transition comparison for $s\in\{0,1,2,4\}$ and $0.001\leq r\leq0.25$: (a) numerical counts and the full two-term expression; (b) the absolute relative difference $|k_{\mathrm{num}}-k_{\mathrm{theory}}|/k_{\mathrm{num}}$, shown as a percentage. Colors identify $s$. Numerical counts are verified first crossings, including score equality.}
  \label{fig:classical_quantum_asymptotic_transition}
\end{figure}

For testing the asymptotic result of Theorem~\ref{th:classical_quantum_transition}, we calculate $56$ transitions for $s\in\{0,1,2,4\}$ and $14$ ratios over $0.001\leq r\leq0.25$. In this range, Theorem~\ref{th:phase_score_bound} reduces the optimization to Grover instances. We verify each first score crossing. The constants $C_s$ and $B_s$ are calculated independently from the limiting posterior integral.

In Fig.~\ref{fig:classical_quantum_asymptotic_transition}(a), we compare these numerical counts with the full expression $k_{\mathrm{asy}}=C_s^2/r^2-2B_s/r$. We observe that the relative difference decreases as the testing ratio decreases. At $r=0.10$, the absolute relative error is between $0.62\%$ and $2.02\%$ across the four values of $s$, while at $r=0.001$ it is between $0.0000403\%$ and $0.000213\%$. Panel~(b) shows the absolute relative difference $|k_{\mathrm{tr}}-k_{\mathrm{asy}}|/k_{\mathrm{tr}}$ on a logarithmic scale. The integer choices of depth and transition count produce fluctuations in this difference, while their scale relative to $k_{\mathrm{tr}}$ vanishes as $r$ decreases. This behavior is consistent with the $\mathcal O_s(1)$ remainder predicted for each fixed $s$.

\section{Conclusion}

\begingroup
\widowpenalty=10000

In this paper we analyzed the influence of the time ratio $r$ of classically testing an element compared to quantum amplitude amplification, for an adaptive Bayesian strategy for complete unstructured search with an unknown number of solutions. The proposed algorithm combines classical sampling, Grover search, and YLC fixed-point search, updating the posterior distribution over the number of remaining solutions after every instance. Each tested element is then removed from the active search space, and the next instance maximizes the posterior-predictive success probability per execution time cost. The search stops when the posterior expected number of undiscovered solutions is at most the prescribed tolerance.

For fixed $r$ and known number of solutions much lower than the number of elements in the database, the greedy policy is asymptotically optimal within the considered search model. The total cost of this policy is approximately $12.14\%$ lower than that of using an exact searching instance at every stage. Thus, requiring certainty at each instance does not minimize the expected cost of complete search in this regime. This result, where the number of solutions is known, also serves as a reference for the simulations with unknown number of solutions, where the algorithm must infer when to stop.

For every posterior distribution and regimes where $0<r\leq1/4$, we have proven that YLC fixed-point search cannot improve the best greedy score obtained from a hybrid approach combining classical sampling and unmodified Grover search. Thus, the optimization over the YLC width can therefore be omitted in this regime. For the $1/4<r\leq1$ regime, we derived an upper bound on the relative score improvement which decreases with the number of amplification iterations. This also represents an upper bound on the additional cost of a generalized iteration such that an advantage remains possible. These statements concern individual searching instance scores, not the total runtime of the resulting adaptive policies.

Starting from a uniform prior, we derived a threshold for optimal transition from classical to quantum search by approximating the probability distribution of the number of solutions. For $1/4<r\leq1$, we determine that the first transition from classical to quantum search instances is attained by an optimized instance with one amplification iteration and we determine the value of the YLC width selected at transition. For a fixed number of observed solutions and $r<1/4$, the number of preceding classical queries has an $r^{-2}$ leading term and an explicit $r^{-1}$ correction.

The findings of this work are fully supported by decision-tree calculations and paired Monte Carlo simulations. Allowing fixed-point searching instances reduces the mean scaled runtime in many of the tested sparse cases at larger verification ratios, but the improvement is not uniform over the solution counts considered. We found that increasing the stopping tolerance can lower the runtime while substantially increasing the fraction of runs that leave solutions undiscovered.

\endgroup

\section*{Author contributions}
Conceptualization: R.A.G. and P.G.P.; methodology: R.A.G. and P.G.P.; validation: R.A.G, and B.C.C.; writing— original draft preparation: R.A.G., B.C.C., and P.G.P.; writing—review and editing: R.A.G., B.C.C., and P.G.P.; supervision: P.G.P.; All authors have read and agreed to the published version of the manuscript.

\section*{Data availability}
The data supporting the findings of this study are available from the corresponding author upon reasonable request.

\section*{Acknowledgment}
This work has been partially supported by RExQTCS: Romanian Excellence in Quantum Technologies enhancing Cybersecurity, a grant of the Ministry of Education and Research from Romania, CCCDI - UEFISCDI, project number PN-IV-P6-6.1-CoEx-2024-0214, within PNCDI IV.


\bibliographystyle{quantum}
\bibliography{sn-bibliography}
\end{document}